\documentclass[journal,12pt]{IEEEtran}

\usepackage{cite}
\usepackage{amsmath,amssymb,amsfonts}
\usepackage{algorithmic}
\usepackage{graphicx}
\usepackage{xfrac}
\usepackage{textcomp}
\usepackage{bm}
\usepackage{xcolor}

\graphicspath{{graphics/}}

\def\BibTeX{{\rm B\kern-.05em{\sc i\kern-.025em b}\kern-.08em
    T\kern-.1667em\lower.7ex\hbox{E}\kern-.125emX}}

\usepackage{hyperref}
\usepackage{verbatim} 

\usepackage{graphicx}
\usepackage{caption}
\usepackage{subcaption}
\usepackage{amssymb}
\usepackage{ textcomp } 
\usepackage{color}
\usepackage{enumerate}
\usepackage{mathrsfs}
\usepackage{multicol}
\usepackage{hyperref}
\usepackage{amsopn}

\usepackage{amsmath}
\usepackage{amsthm}
\usepackage{amssymb}

\usepackage{multicol}

\usepackage{algorithm}
\usepackage{algorithmic} 

\newcommand{\AAA}{\mathbf{A}}
\renewcommand{\AA}{\mathbf{A}}
\newcommand{\BB}{\mathbf{B}}
\newcommand{\CC}{\mathbf{C}}
\newcommand{\DD}{\mathbf{D}}
\newcommand{\II}{\mathbf{I}}
\newcommand{\MM}{\mathbf{M}}
\newcommand{\PPhi}{\mathbf{\Phi}}
\newcommand{\GGamma}{\mathbf{\Gamma}}
\newcommand{\DDelta}{\mathbf{\Delta}}
\newcommand{\aalpha}{\mathbf{\alpha}}

\newcommand{\aaa}{\mathbf{a}}
\renewcommand{\aa}{\mathbf{a}}
\newcommand{\bb}{\mathbf{b}}
\newcommand{\dd}{\mathbf{d}}
\newcommand{\ee}{\mathbf{e}}
\newcommand{\eeta}{\mathbf{\eta}}
\newcommand{\hh}{\mathbf{h}}
\newcommand{\pp}{\mathbf{p}}

\newcommand{\uu}{\mathbf{u}}
\newcommand{\vv}{\mathbf{v}}
\newcommand{\ww}{\mathbf{w}}
\newcommand{\xx}{\mathbf{x}}
\newcommand{\yy}{\mathbf{y}}
\newcommand{\zz}{\mathbf{z}}

\newcommand{\C}{\mathbb{C}}

\newcommand{\R}{\mathbb{R}}

\newcommand{\calB}{\mathcal{B}}

\newcommand{\osigma}{\overline{\sigma}}
\newcommand{\otau}{\overline{\tau}}

\newcommand{\nicebinom}{\genfrac{(}{)}{0pt}{1}{s}{k_s}}
\newcommand{\abs}[1]{\left\vert #1 \right\vert}
\newcommand{\norm}[1]{\Vert #1 \Vert}

\newcommand{\set}[1]{\left\lbrace #1\right\rbrace}
\newcommand{\sse}{\subseteq}
\newcommand{\sprod}[1]{\left\langle #1 \right\rangle}

\newcommand{\prb}[1]{\mathbb{P}\left( #1 \right)}
\newcommand{\erw}[1]{\mathbb{E}\left( #1 \right)}

\usepackage{dsfont}

\newcommand{\geqsim}{\gtrsim}
\newcommand{\leqsim}{\lesssim}

\DeclareMathOperator{\supp}{supp}

\DeclareMathOperator{\diag}{diag}

\DeclareMathOperator{\tr}{tr}

\theoremstyle{definition}

\numberwithin{lem}{section}

\usepackage{etoolbox}
\let\bbordermatrix\bordermatrix
\patchcmd{\bbordermatrix}{8.75}{4.75}{}{}
\patchcmd{\bbordermatrix}{\left(}{\left[}{}{}
\patchcmd{\bbordermatrix}{\right)}{\right]}{}{}

\newtheorem{theorem}{Theorem}
\newtheorem{remark}[theorem]{Remark}

\newtheorem{corollary}[theorem]{Corollary}
\newtheorem{lemma}[theorem]{Lemma}
\newtheorem{proposition}[theorem]{Proposition}

\title{One-Shot Messaging at Any Load Through Random Sub-Channeling in OFDM}

\author{    \IEEEauthorblockN{Gerhard Wunder,\IEEEauthorrefmark{1}, Axel Flinth\IEEEauthorrefmark{2}, Benedikt Gro\ss \IEEEauthorrefmark{1}}

    \IEEEauthorblockA{\IEEEauthorrefmark{1}Department of Computer Science, 
Freie Universit\"at Berlin, Germany}

    \IEEEauthorblockA{\IEEEauthorrefmark{2} Department of Mathematics and Mathematical Statistics, Umeå University, Sweden}
\thanks{
$\copyright$ 2023 IEEE. Personal use of this material is permitted. Permission from IEEE must be obtained for all other uses, in any current or future media, including reprinting/republishing this material for advertising or promotional purposes, creating new collective works, for resale or redistribution to servers or lists, or reuse of any copyrighted component of this work in other works.

Accepted for publication in IEEE Transactions on Information Theory. DOI: 10.1109/TIT.2023.3283063 .

This paper was presented in parts in IEEE Statistical Signal Processing Workshop \cite{wunder2021measure}.}
}

\begin{document}

\maketitle

\begin{abstract}
Compressive Sensing (CS) has well boosted massive random access protocols over the last decade.
Usually, on physical layer, the protocols employ some \emph{fat matrix} with the property that sparse vectors in the much larger column space domain can still be recovered. This, in turn, greatly reduces the chances of collisions between access devices. This basic scheme has meanwhile been enhanced in various directions but the system cannot operate in overload regime, i.e. sustain significantly more users than the row dimension of the fat matrix dictates.
In this paper, we take a different route and apply an orthogonal DFT basis as it is used in OFDM, but subdivide its image into so-called sub-channels and let each sub-channel take only a fraction of the load. In a random fashion the subdivision is consecutively applied over a suitable number of time-slots. Within the time-slots the users will not change their sub-channel assignment and send in parallel the data. Activity detection is carried out jointly across time-slots in each of the sub-channels. For such system design we derive three rather fundamental results:
i) First, we prove that the subdivision can be driven to the extent that the activity in each sub-channel is sparse by design. An effect that we call \emph{sparsity capture effect}.
ii) Second, we prove that effectively the system can sustain any \emph{overload situation} relative to the DFT dimension, i.e. detection failure of active and non-active users can be kept below any desired threshold regardless of the number of users. The only price to pay is delay, i.e. the number of time-slots over which
cross-detection is performed. We achieve this by jointly exploring the effect of \emph{measure concentration} in time and frequency and careful system parameter scaling. 
iii) Third, we prove that parallel to activity detection active users can carry one symbol per pilot \ and time-slot so it supports so-called \emph{one-shot messaging}.
The key to proving these results are new concentration results for sequences of randomly sub-sampled {DFT}s detecting the sparse vectors "en bloc". 
Eventually, we show by simulations that the system is
scalable resulting in a coarsely 20-fold capacity increase
compared to standard OFDM.
\end{abstract}

\section{Introduction}


\IEEEPARstart{T}{here} is meanwhile an unmanageable body of literature on CS for the (massive)
random access channel (RACH) in wireless networks, often termed as compressive random access
\cite{Wunder2015_ACCESS,Bockelm2016_COMMAG}. Zhu et al \cite{Zhu2011_TWC} and later Applebaum et
al. \cite{Applebaum2012_PHYCOM} were the first to recognize the benefit of
sparsity in multiuser detection, followed up by a series of works by
Bockelmann et al. \cite{Bockelm2013_ETT,Yi2014_GC} and recently by Choi
\cite{Choi2017_IoT,Choi2017_TWC}. A single-stage, grant-free (i.e. one-shot)
approach has been proposed in
\cite{Wunder2014_ICC,Wunder2015_GC,Wunder2015_ASILOMAR}
where both data and pilot channels are overloaded within the same OFDM symbol.
A new class of hierarchical CS (h-CS) algorithm tailored for this problem has
been introduced in \cite{BaraniukEtAl:2010,SprechmannEtAl:2011} (for
LASSO)\ \cite{Schepker2013_ISWCS} (for OMP) and recently in
\cite{Wunder2017_ASILOMAR,Roth2020_TSP} (for HTP and Kronecker measurements). A comprehensive
overview of competitive approaches within 5G can be found in
\cite{Bockelm2018_ACCESS}. Recently, a surge of papers has
combined RACH system design with massive MIMO which adds another
design parameter (number of antennas) to the problem, see e.g.
\cite{Liu2018-I_TSP,Liu2018-II_TSP,Carvalho2017_TWC,Wunder2019_TWC}.
The information-theoretic link between random access and CS, i.e. leveraging the use of a common codebook, has been explored in \cite{Polyanskiy2017_ISIT,Yu2017_ITA}. This has been taken forward in many works, see \cite{Kowshik2020_TC,Amalladinne2020_TIT,Yavas2021_TIT,Fengler2021_TIT,Kowshik21_ARXIV}.
Notably, CS together particularly with OFDM still plays a key role in upcoming 6G RACH design \cite{Wu2020_TWC}, see, e.g., the recent work by Fengler et al \cite{Fengler2022_SAC}.

The very recent papers by Choi \cite{Choi2018_TVC}\cite{Choi2020_IoT}, brought to our attention by the author, have revived our interest in the RACH design problem. In
\cite{Choi2018_TVC} a two-stage, grant-free approach has been presented.
In the first stage, a classical CS-based detector detects the active $n$-dimensional
pilots from a large set of size $r>n$. The second stage consists of data
transmission using the pilots as spreading sequences. \cite{Choi2020_IoT} has
presented an improved version where the data slots are granted through prior
feedback. The throughput is analysed and simulations show significant
improvement over multi-channel ALOHA. However, by design the scheme cannot be overloaded (see equation (3) in \cite{Choi2020_IoT}). Missed detection analysis is
carried out under overly optimistic assumptions, such as ML detection, making
the results fragile (e.g. the missed detection cannot be independent of $r$ as
the results in \cite{Choi2020_IoT} suggest). Moreover no concrete pilot design
(just random) and no frequency diversity is considered which is crucial for
the applicability of the design. So, the achievable load scaling of this scheme remains unclear.

We take a different approach here:\ Instead of overloading $n$ compressive measurements
with $r>n$ pilots we use $n$-point DFT (orthogonal basis) and subdivide the available bandwidth into sub-channels, each of which serving a few of the pilots. We send exclusive pilot symbols in the first time-slot only and data in the remaining time-slots. Then, we apply hierarchical CS algorithms for joint activity detection over a number of time-slots in each of the sub-channels. Notably, the system is reminiscent of an OFDM system where the sub-channels correspond to bundled sub-carriers and the time-slot to a sequence of OFDM symbols. For this system, in essence, we provide a theoretical guarantee that the system can sustain any overload situation, as long as the other design parameters, e.g., number of pilots and time-slots, are appropriately scaled. Overload operation means that many more than $n$ users, where $n$ is the 
signal space dimension dictated by the DFT, can be reliably detected and each of which can carry one data symbol per pilot dimension and time-slot, i.e. the system supports \emph{one-shot messaging}. The only price that we pay is delay, i.e. the number of time-slots might have to be adapted.
This is achieved by roughly bundling $\log(n)$ sub-carriers for a pool of (only) $\log^2(n)$ pilots over $\log^4(n)\log(\log(n))$ time-slots, which are then collaboratively detected. Not only will this scaling entail sparsity in each of the sub-channels by design, so-called \emph{sparsity capture effect}, but still allow reliable detection by exploiting the joint detection over time-slots.
The main tool is to establish new concentration results for a family of vectors with common sparsity pattern.
Technically, our analysis rests entirely upon utilizing the \emph{mutual coherence} properties of the DFT matrices instead of more sophisticated methods, such as the \emph{restricted isometry property} (RIP) \cite{FouRau2013}.
In the simulation section, this is validated for several system settings yielding a 20-fold increase in user capacity.

The paper is organized as follows: In Section 2 we will introduce the system model in great detail. The sparsity capture effect is analysed in Section III. In Section IV the detection performance is analysed. Data recovery is analysed in Section V. Numerical experiments are provided in Section VI. An overview of our notation can be found in Table \ref{tab:notation}.
\begin{table}
    \centering
    \begin{tabular}{c|c}
        $n$ & Dimensions of the signatures $p_\ell$ \\
        $r$ & Number of possible signatures per sub-channel \\
        $\hh_k'$ & Channel impulse response (CIR) of the $k$:th user \\
        $s$ & Maximum length of the CIRs\\
        $\hh_{\ell,j}^i$ & Effective channels, $\ell \in [r]$ \\
        $k_s$ & Sparsity of the CIRs \\
         {$k_u$}&  {High-probability bound of active users per sub-channel. } \\
        $\calB$ & Sub-carriers within a sub-channel \\
        $\PPhi_\calB$ & Sub-channel sampling operator  \\
        $m$ & Number of sub-carriers per sub-channel \\
        $c$ & Number of sub-channels, $n = c \cdot m$ 
    \end{tabular}
    \caption{Nomenclature}
    \label{tab:notation}
    
\end{table}

\section{Model description}
{  \label{sec:system}
We imagine a set of $u$ users $k \in [u]:={0,1, \dots, u-1}$, that are communicating with a base station over $t$ time-slots $i \in [t]$ in an totally uncoordinated fashion. We assume an OFDM-like system, i.e. with some cyclic prefix, operating with IDFT/DFT matrix of size $n \times n$ where $n$ is the signal space dimension. The time-slots correspond then to OFDM symbols. The ultimate objective for the users is to be reliably detected and to transmit one datum $d_k^i \in \C$ per time-slot at the same time, i.e. one-shot messaging. Time-slot $0$ is reserved for multi-path channel estimation only, so that $d_k^0=1$ for all $k$. The users transmit their data by modulating a set of pilots $\pp_k^i \in \C^n$, $i \in [t]$. Let $\hh_k' \in \C^s$ denote the (sampled) channel impulse response (CIR) of the $k$:th active user, where $s\ll n$ is the length of the cyclic prefix, which is assumed to be constant over all time-slots. Inactive users are modeled by $h_{k}^{\prime}=0$. Hence, the base station will receive for $i=0,...,t-1$
\begin{align*}
    \yy^i := \sum_{ k \in [u]} \pp_{k}^i*d_k^i\hh_k'.
\end{align*}
The users choose their pilots $\pp_k^i$, $i \in [t]$, as follows. At the beginning of the $t$ time-slots, each  user $k$ chooses two indices $j \in [c]$ and $\ell \in [r]$. Notably, this choice determines their pilot choices for all time-slots. We will refer to $j$ as the \emph{sub-channel index}, and $\ell$ as the \emph{pilot index}. We assume that $rs\leq n$. Accordingly, to each sub-channel, we associate a sequence of sets $(\calB_{j}^i)_{i \in [t]}$, $\calB_j^i \sse [n], i \in [t]$. The $\calB_j^i$ all have the size $m$ (so that $c=n/m$), and are for each time-slot $i$ disjoint: $\calB_j^i \cap \calB_{j'}^i = \emptyset$ for $j\neq j'$. As we will see later in Section \ref{sec:pilots}, the $\calB_j^i$ correspond to a sub-division of the DFT sub-carriers $[n]$ in frequency domain into $c$ \emph{sub-channels}, hence the name sub-channel index. The $\calB_j^i$, $j \in [c]$, will one by one be chosen disjoint to all previously drawn subsets, but otherwise uniformly at random. We shall consider both the fixed case $\calB_j^i=\calB_j^0$ and the case where this procedure is repeated independently for each $i$. This will be done by the base-station, which will then broadcast the selection to the users.
Let us denote the (random) mapping from users to index pairs by $k \hookrightarrow (j,\ell)$. We then have
\begin{align*}
     \sum_{ k \in [u]} \pp_{k}^i*d_k^i\hh_k' &= \sum_{j, \ell \in [c]\times [r]}\sum_{k \hookrightarrow (j,\ell)}\pp_{\ell,j}^i * d_k^i\hh_k' \\
     &= \sum_{j, \ell \in [c]\times [r]}\mathrm{circ}^{(s)}   (\pp_{\ell,j}^i )\hh_{\ell,j}^i, 
\end{align*}
where $\mathrm{circ}^{(s)}(\pp)$ denotes the circular matrix in $\C^{n,s}$ defined by vector $\pp \in \C^n$. Also,  we have defined \emph{effective channels}
\begin{align*}
    \hh_{\ell,j}^i = \sum_{k \hookrightarrow (j, \ell)} d_k^i \hh_k'.
\end{align*}
Here we have included the data as a part of the effective channels for ease of notation.
Next, let us stack the circular matrices $\mathrm{circ}^{(s)}(\pp_{\ell,j}^i ) \in \C^{n,s}$ and effective CIRs $\hh_{\ell,j}^i \in  \C^s$ into matrices $\CC(\pp_{\cdot,j}^i) = [\mathrm{circ}^{(s)}(\pp_{0,j}^i) , \dots, \mathrm{circ}^{(s)}(\pp_{r,j}^i)] \in \C^{n,rs}$ and vectors $\hh_j^i \in \C^{rs}$. By possibly concatenating zero elements into them, we may think of them as matrices in $\C^{n,n}$ and vectors $\hh^i_j \in \C^n$.
Overall, the signal the base station receives in time slot $i$ is given by
\begin{align*}
    \yy^i = \sum_{j \in [c]} \CC(\pp_{\cdot,j}^i)\hh_j^i + \ee^i,
\end{align*}
where $\ee^i \sim \mathcal{CN}(0, \sigma^2\II_n)$ is the white noise. A schematic description of the proposed scheme is given in Figure \ref{fig:system}.
\begin{figure*}
    \centering
    \includegraphics[width = \textwidth]{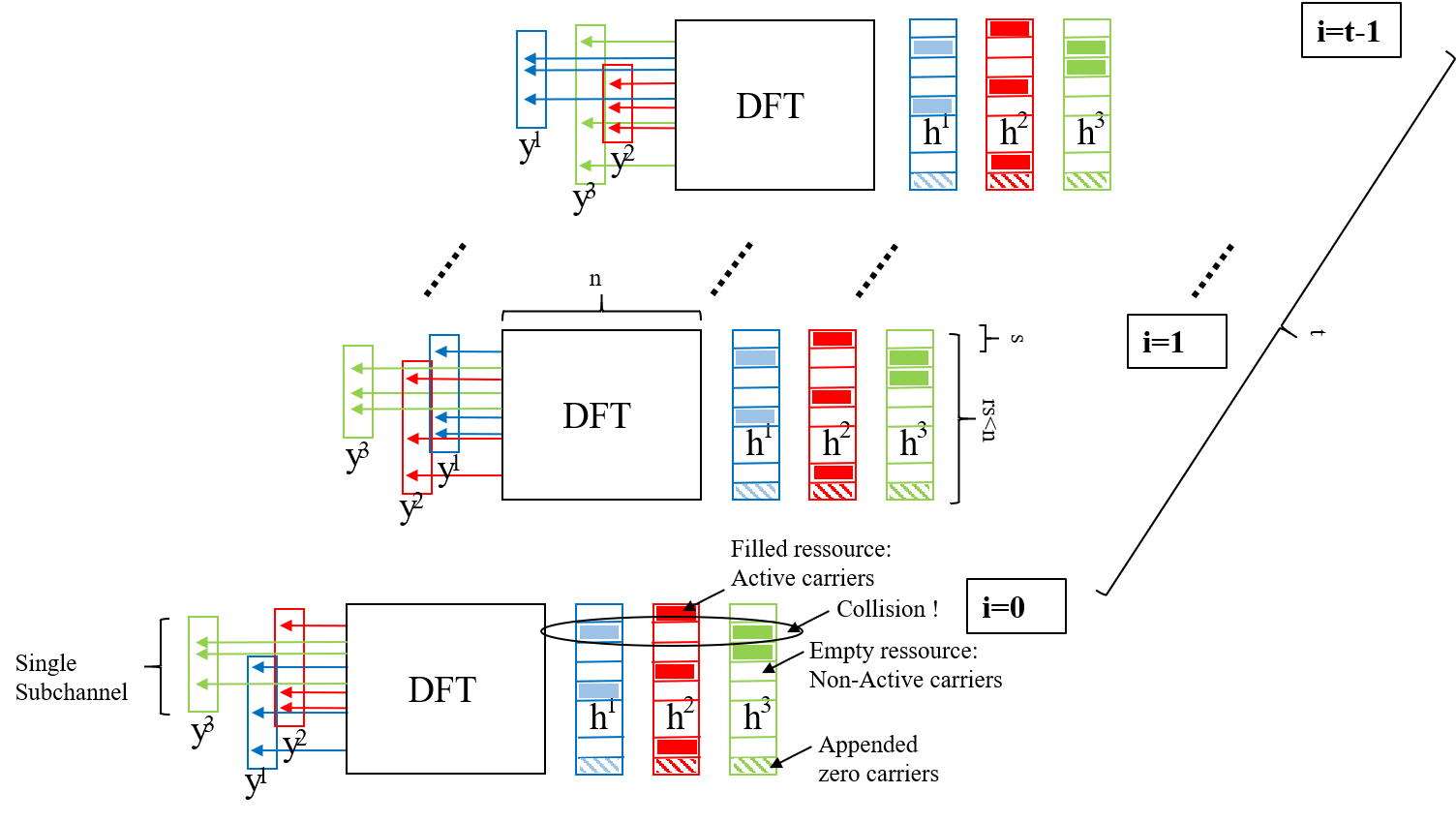}
    \caption{An illustration of the proposed OFDM-like RACH with randomly varying sub-channels in frequency domain (sub-carriers) and time domain resources with common block support over the time-slots (OFDM symbols).}
    \label{fig:system}
\end{figure*}


\subsection{Proxy measurement model} \label{sec:pilots}
We shall analyse a particular choice of the pilots in each time-slot. To each sub-channel (and time-slot), we associate  $r$ pilots $\pp_{\ell,j}^i$. These are constructed as follows: first, a 'base pilot' $\pp_{0,j}^i$ is chosen as a vector with DFT supported on $\calB_j^i$ with constant unit power. More concretely,
\begin{align} \label{eq:pilots}
    \abs{\hat{\pp}_{0,j,p}^i} = \begin{cases}  \sqrt{\tfrac{n}{m}} & p \in \calB_j^i \\
     0 & \text{ else.}\end{cases}
\end{align}
The other $(r-1)$ pilots are defined as cyclical shifts of the base pilots
\begin{align*}
    \pp_{\ell,j}^i = (\pp_{0,j}^i)^{(\ell s)}.
\end{align*}
Due to the duality of modulation in DFT domain and translation in time domain, all $p_{\ell,j}^i$, $\ell \in [r]$ are supported in DFT domain on the set $\calB_j^i$.
Now, importantly, the structure of the $\pp_{\ell,j}^i$ in \eqref{eq:pilots} implies that each $\CC(\pp_{\cdot,j}^i)$ has the structure of a circular matrix:
\begin{align}
    \CC(\pp_{\cdot,j}^i) = \mathrm{circ}^{(n)}(\pp_{0,j}^i). \label{eq:Cp}
\end{align}

A key idea in CS is to perform the
user identification and channel estimation task within a linear subspace of much smaller dimension $m \ll n$. We propose for the base station to take such \emph{compressive measurements} as follows: Let $\PPhi \in \C^{n,n}$ be the normalized DFT matrix, $\Phi_{pq} = n^{-\sfrac{1}{2}} e^{-\sfrac{2\pi \iota pq}{n}}$, $p,q \in [n]$. The \emph{$j'$:th compressive measurement} in time-slot $i$, $j'\in [c]$, is then given by
\begin{align*}
    \bb_{j'}^i = \PPhi_{\calB_{j'}^i}\yy^i \in \C^m,
\end{align*}
where $\calB_{j'}^i$ are the subsets defined in the last section. Now, the circular structure of the $\CC(\pp_{\cdot,j}^i)$ \eqref{eq:Cp} together with their simultaneous diagonalisation property implies that $\CC(\pp_{\cdot,j}^i)= \PPhi^*\diag(\sqrt{n}\hat{\pp}_{0,j}^i) \PPhi$. Consequently,
\begin{align*}
       \bb_{j'}^i &= \sum_{j \in [c]} \PPhi_{\calB_{j'}^i}\PPhi^*\diag(\sqrt{n}\hat{\pp}_{0,j}^i) \PPhi\hh_j^i + \PPhi_{\calB_{j'}^i}\ee^i  \\
                 &= \sum_{j \in [c]}\diag(\sqrt{n}(\hat{\pp}_{0,j}^i)_{\calB_{j'}^i}) \PPhi\hh_j^i + \PPhi_{\calB_{j'}^i}\ee^i.
\end{align*}
Due to the $\calB_j^i$ being disjoint, we have $\left.\hat{\pp}_{0,j}^i\right|_{\calB_{j'}^i}=\mathbf{0}$ for $j\neq j'$. Hence, the measurement $\bb_{j'}^i$ only depends on the effective channel $\hh_{j'}^i$ associated to the sub-channel $j'$. Disregarding the phases of $\hat{\pp}_{0,j}^i$ (which are not important for the following analysis) and renormalising, we can write it as 
\begin{align*}
    \bb^i_j = \AAA^i_j (\hh^i_j + \zz^i_j),
\end{align*}
where $\AAA^i_j = \sqrt{\tfrac{n}{m}}\PPhi_{\calB^i_j} \in \C^{m,n}$ is an (average) energy-preserving, sub-sampled version of the DFT matrix, and $\zz^i_j \in \C^n$ is Gaussian with zero mean and covariance matrix $\tfrac{\sigma^2m}{n^2}\II_n$. Hence, $\bb_j^i$ is a low-dimensional image of the effective channel $\hh_j^i$ associated to the index  $j$. We will in the following drop the sub-channel index $j$, since the sub-channels clearly can be processed completely in parallel.

Note that expressing the noise as $\zz^i$ instead of $\ee^i$ has an formal regularizing effect -- the variance of the entries of the $\zz^i$ are smaller than the ones of $\ee^i$. We will use this heavily in our proofs. Also note that since only coefficients corresponding to the first $rs$ members of $[n]$ are active in $\hh^i$, we may think of the $\AAA^i$ as operators in $\C^{m,rs}$ instead of $\C^{n,n}$. We will use this fact at some critical steps of our argument in the analysis.

\subsection{Hierarchical sparsity (by design)}
So far we have not assumed any sparsity of the vectors $\hh^i_j$ at all, although this is a fundamental prerequisite of any CS detection algorithm. In fact, we will not explicitly assume sparsity anywhere in this paper, but instead show in the analysis later that each sub-channel will become essentially \emph{sparse by design}. To be precise, they will become sparse in a more generalized sense, so-called \emph{hierarchically sparse}. {Let us shortly define hierarchical sparsity. As it is well known, a vector $\xx \in \C^n$ is \emph{$s$}-sparse if at most $s$ coefficients $x_i$ are non-zero. A vector $\C^{bs}\ni \xx= [\xx_0, \dots \xx_{b-1}]$ consisting of $b$ blocks $\xx_k \in \C^s$ is likewise called \emph{$(\kappa,\sigma)$-hierarchically sparse} if at most $\kappa$ of the $\xx_i$ are non-zero, and additionally each $\xx_i$ in itself is $\sigma$-sparse. Iteratively, one can define hierarchical sparsity of any number of levels $(s_1, s_2, \dots, s_L)$. For a more thorough introduction, we refer to \cite{cosipChapter}.}

We will give a detailed analysis of the sparsity pattern of the effective channels below. To already now give some intuition what happens, recall that the $\hh^i$'s are composed of $r$ active or non-active blocks of length $s$. Each block corresponds to a pilot index. First, clearly, the individual CIRs $\hh_k'$ can be interpreted as $k_s$-sparse (i.e., including the special case $k_s=s$). 
If the number of channels $c$ 
is appropriately scaled, 
 the users will distribute approximately equally between the sub-channels. Hence, each specific user will not compete with significantly more than $\tfrac{u}{c}=\tfrac{m}{n}u$ other users -- or in mathematical terms, the number of active blocks in the vector $\hh^{i}$ will not be not significantly more than $\tfrac{m}{n}u$ with high probability. Concretely, as we show in Section \ref{sec:sparsity}, the number of nonzero blocks is with high probability bounded by
\begin{equation} \label{eq:sparsity}
    k_u := 2\tfrac{m}{n}u.
\end{equation}
Now, if in addition 
the overall number of blocks $r$ (i.e. pilot indices)  per sub-channel is large enough, all users falling within the same sub-channel will choose different pilots with high probability. Interestingly, it will turn out that the right scaling for $r$ is sub-linear in $n$ (as shown in the main result in Section \ref{sec:detection}). Consequently, the blocks will still be $k_s$-sparse. Importantly, all $(\hh^{i})_{i \in [t]}$ have the same support. This fact is something we will heavily use in our analysis. The common support enables the detection of $k^i$ users with less measurements than in standard CS. Altogether, { $\hh = (\hh^0, \dots, \hh^{(t-1)})$ can be viewed as a $(k_u,k_s,t)$-sparse vector.}


\subsection{A detection algorithm}
The recovery of hierarchically sparse vectors has been extensively studied in a series of papers by the authors of this article. We refer to \cite{cosipChapter} for an overview. One of the main findings is that the so-called HiHTP-algorithm can be used to recover them efficiently. The algorithm is in essence a projected gradient descent, whereby projection refers to projection onto the set of hierarchically sparse vector. 

This projection can be calculated efficiently using the principle of optimal substructures. As an example, to calculate the best $(\kappa,\sigma)$-sparse approximation of a vector $(\xx_0, \dots, \xx_{b-1})$, we first calculate the best $\sigma$-sparse approximation $\hat{\xx}_k$ of each block $\xx_k$, and subsequently choose the $\kappa$ values of $k$ for which $\norm{\hat{\xx}_k}$ are the largest. This strategy carries through to more levels of sparsity. 

In this paper, we will show that one step of the HiHTP algorithm can be used to detect the users and their data in each channel.  Concretely, given a set of measurements $(\bb^i)_{i \in [t]}$, we for each sub-channel $j$ calculate the vector $((\AAA^i)^*\bb^i)_{i \in [t]} \in \C^{trs}$, project it onto the set of $(k_u,k_s,t)$-sparse vectors, and subsequently solve a least-squares problem restricted to the support of the projection. Written out, this means:

\begin{enumerate}
    \item For each $k$, determine the $k_s$ values of $\ell$ for which
    \begin{align*}
        \sum_{i \in [t]} \abs{((\AAA^i)^*\bb^i)_{(k,\ell)}}^2
    \end{align*}
    is the largest. Declare that set as $\omega_k$.
    \item Determine the $k_u$ values of $k$ for which
    \begin{align*}
        \sum_{\ell \in \omega_k} \sum_{i \in [t]} \abs{((\AAA^i)^*\bb^i)}_{(k,\ell)}^2
    \end{align*}
    are the largest, while still being larger than some threshold $\vartheta>0$. Call this set $\mathcal{I}$.
    \item The sets $(\omega_k)_{k \in I}$ are then fused together in $\Omega = \set{(k,\ell) \, \vert \, k \in \mathcal{I}, \ell \in \omega_k}$.
    \item Determine solutions $\hh^i_*$ to the least-squares problems restricted to $\Omega$
    \begin{align*}
        \min_{\substack{h^i, \supp \hh^i \sse \Omega}} \norm{\bb^i - \AAA^i \hh^i}^2.
    \end{align*}
    \item Calculate estimates $\dd^i_*$  of the data vectors by calculating the element-wise quotient $\tfrac{\hh_*^i}{\hh_*^0}$.
\end{enumerate}
In the next sections, we will analyse this detection algorithm. {The argument proceeds in two steps: We first prove in Section \ref{sec:sparsity} that the $\hh_k$ are hierarchically sparse with high probability. With this in mind, we prove in the following Section \ref{sec:detection} that the users in any group are correctly classified with high probability.}

}

\section{Sparsity capture effect} \label{sec:sparsity}
Let us begin by carrying out the argument sketched in the introduction that the $\hh^i$ are, with high probability, hierarchically sparse.
\begin{proposition} \label{prop:collision}
    For each sub-channel and $\lambda>0$, the probability that there are more than 
    $$(1+\lambda) \tfrac{m}{n}u$$ 
    users in the sub-channel is smaller than $\exp\left(- \tfrac{3\lambda^2 mu}{n(1 + 3\lambda )}\right)$.
\end{proposition}
\begin{proof}
    Let $X_i$, $i \in [u]$ be random variables which are equal to $1$ if user $i$ is in the sub-channel, and zero otherwise. Of course, $X_i$ is $\mathrm{Ber}(p)$ distributed, where we for convenience defined $p=\tfrac{m}{n}$. Therefore, $\mathbb{E}(X_i)=p$, and $\mathbb{V}(X_i)= p(1-p)$. Furthermore, we  have $\abs{X_i-p}\leq 1$ almost surely. The Bernstein inequality \cite[Th. 2.8.1]{vershynin2010introduction} therefore implies that
    \begin{align*}
       \mathbb{P}\bigg( \sum_{i \in [u]} X_i - pu>\lambda \bigg)\leq \exp\bigg(\frac{-\sfrac{\lambda^2}{2}}{u p(1-p) + \tfrac{\lambda}{3}}\bigg)
    \end{align*}
    Now set $\lambda = \lambda pu$ and estimate $p(1-p)\leq p$, to get the result.
\end{proof}
Hence, when analysing a specific channel, we may with very high probability assume that only {$k_u$}
of the potential $u$ CIR's $h_k'$ are non-zero, where $k_u$ is as defined in \eqref{eq:sparsity}{} (which corresponds to $\lambda=1$). The same is true for the $\hh_k$, since the $k_u$ users choose at most $k_u$ different pilots. We can however prove more.
\begin{proposition} \label{prop:nocollisions}
     Fix a channel. The probability of a collision {in a sub-channel}, i.e., two users choosing the same pilot, conditioned on the event that there are no more than $k_u$ users,  is smaller than
    \begin{align*}
        \frac{k_u^2}{2 r}
    \end{align*}
\end{proposition}
\begin{proof}
     What we are dealing with is clearly a 'birthday paradox problem' with $k_u$ objects being distributed in $r$ bins. For any given pair of objects, probability of a collision is clearly equal to $r^{-1}$. A union bound over the number of pairs now gives the claim.
\end{proof}

\begin{remark}
    This bound is somewhat pessimistic, but not much so.  It is not hard to prove that the probability of at least one collision happening is bigger than $1-\exp(-\tfrac{k_u^2}{2r})$. In the following, we will choose the parameters in a way that ensures that $k_u^2/r$ is small. In this regime, $1-\exp(-\tfrac{k_u^2}{2r})$ is very close to the above value. 
\end{remark}

In the event that no collisions occur, each effective CIR $\hh_\ell$ is equal to at most one single $\hh_k'$. Hence, in this case, at most $k_u$ of the $u$ blocks $\hh_\ell$ in $\hh$ are non-zero. Furthermore, each of these are the $k_s$-sparse by assumption. Putting things together, we obtain the following result.

\begin{theorem} \label{theo:sparsity}
Fix a channel. With a failure probability smaller than 
\begin{align*}
\exp\left(- \tfrac{3mu}{4n}\right) +\tfrac{4m^2u^2}{n^2r},
\end{align*}
the stacked vector of effective CIR's $\hh$ is $(k_u,k_s)$-sparse, with $k_u$ defined in \eqref{eq:sparsity}, and $\hh_i = \diag(\dd_i)\hh$, where $\dd_i \in \C^n$ are the stacked data of all users.
\end{theorem}

\section{Detection analysis -- Main result} \label{sec:detection}

We move on to analysing the performance of our recovery algorithm for detecting the correct users. { We will pose the following assumptions on the transmitted data and effective channels.
\begin{quote}
    {\bf Assumption 1:} The data scalars $d_{k}^i\in \C$, $k \in [u], i \in [t]$ are independent. Furthermore, they are independently distributed according to a centered distribution $d$ on the complex unit circle.
\end{quote}
The above assumption is true if the users are sending messages which are uniformly randomly encoded using either a binary or QPSK coding. Likewise, we make the following assumption for the channels:
\begin{quote}
    {\bf Assumption 2:} 
    We {assume that the norms of the $\hh_k'$} are essentially constant. Formally, we assume that 
    \begin{align} \label{eq:normcontrol}
        \forall \, k: \, \tfrac{3}{4} \leq \norm{\hh_k'}^2 \leq \tfrac{5}{4} .
    \end{align}
\end{quote}
Note that the latter is simply a form of \emph{power control} which keeps track of the received energy at the receiver. The absolute values of the constants in \eqref{eq:normcontrol} are  somewhat arbitrary -- it would be possible to carry out the analysis under an assumption of the form $\, \alpha \leq \norm{\hh_k'}^2 \leq \beta$ for any constants $\alpha, \beta > 0 $ -- this would only lead to worse implicit constants. Since the concrete choice of constants ultimately increases readability, we have chosen to do so}.

In what follows, we present our main result. The figure of merit is the user load $u/n$.
We use $\geqsim$ or $\leqsim$ to indicate that inequalities hold up to multiplicative constants independent of all other design parameters, similar to "big O" notation.
{\begin{theorem} \label{th:mainresult}
    Let each user select its sub-channel and pilot independently. Let $\epsilon>0$ be a probability threshold and fix $C_o>0$ and $\kappa>2$. Assume that the noise level and number of {pilot sequences} per sub-channel obey
    \begin{align*}
        r &\geq 8\cdot \left(\tfrac{16\kappa}{3C_0}\right)^2\log(n)^2\epsilon^{-1} \\
        \sigma^2 &\leqsim \tfrac{u}{\log(n)^2}\norm{\hh}^2.
    \end{align*}
     Then, if the sub-channel size $m$ and threshold $\vartheta$ is chosen correctly, and 
    \begin{enumerate}
        \item in the case of all $\AAA^i$ being equal, the overload and acquisition times obey
        \begin{align*}
           \tfrac{n}{u} & \geq C_o k_s^2 \log(n)\\
           \tfrac{t}{\log(t)^2}  & \geqsim (\log(r) + k_s\log(s))^2 \\
        \end{align*}
        \item in the case of the $\AAA^i$ being independently drawn, the overload and acquisition times obey
        \begin{align*}
            \tfrac{n}{u} &\geq C_o k_s^2   \\
            \tfrac{t}{\log(tn)^4} & \geqsim (\log(r) + k_s\log(s))^2
        \end{align*} 
    \end{enumerate}
    the probability that the algorithm will fail to classify the users in a specific sub-channel is smaller than
    \begin{align*}
        \epsilon +n^{-\tfrac{16\kappa}{3C_0}} + n^{2-\kappa} + (tr\nicebinom)^{1-\kappa}.
    \end{align*}
\end{theorem}}

\begin{remark} \label{rem:disclaimer}
The reader should pay close attention to the order of the words here: We do not claim that all users across all sub-channels will be correctly detected with high probability. Instead, we claim that for each user in a sub-channel there is a high probability that the user is correctly detected. In other words: In each transmission period, the base station will probably fail to a detect a few users correctly altogether. However, each sole user will only very infrequently experience not getting properly detected.
\end{remark}

The theorem shows that as long as the number of time-slots $t$ and pilots $r$ grow at a rate $\mathrm{polylog}(n)$, 
the probability that the algorithm will fail to classify the users in a specific sub-channel correctly will for large $n$, correctly scaling with $u$, be very small. The 'correct scaling' depends on how the $\calB_j^i$ are chosen:
\begin{itemize}
    \item If they are drawn once for all time-slots, the number of users that can be
accomodated grows as $\tfrac{n}{k_s^2\log(n)}$, 
\item 
if the $\calB_j^i$ are independent for different time-slots $i$, the number even scales linearly with $n$. 
\end{itemize}
By adjusting values of other constants, this ultimately means (for independent $\calB_j^i$) that it will in theory work at any load. Notice that $n$ and $t$ are design parameters of the algorithm.

For a concrete system design, let us provide a 'cooking recipe' for choosing them:
\begin{enumerate}
    \item The choice of DFT size $n$ in OFDM is typically a trade-off between spectral efficiency and how fast the channel is changing within the OFDM symbol. Let us for simplicity assume that the mobility is not the major limiting part as it is common in massive IoT systems. Nevertheless this sets an upper bound on $n$.
    \item We can choose the maximum expected load in the system by fixing the constant $C_o>0$ in Theorem 11. Notably, $C_o$ can be interpreted as an \emph{inverse overload factor} (with regard to $nk_s^2$). E.g., say $C_o=\tfrac{1}{2}$ means $\tfrac{2n}{k_s^2}$ users are served. Fixing also $\kappa>2$ and $\epsilon$ which both govern the detection failure probability, a lower bound on $n$ is given by 
    \begin{align*}
    n \geq 8s\cdot \left(\tfrac{16\kappa}{3C_0}\right)^2\log(n)^2\epsilon^{-1}
    \end{align*}
    due to the implicit constraint $n\geq rs$.
    
    \item  Given $n,r$ and the channel parameters $s,k_s$ the number of time-slots can be fixed. Unfortunately, the implicit constant in the scaling of $t$ would be very technical to estimate, and would not bring that much insight -- any obtainable bound will in any case be very crude, and we think that it is better to tune $t$ empirically.
\end{enumerate}
Summarizing, the only caveat here to keep the detection failure probability 
below some threshold is that both $n\geq rs$ and $t\geq 1$ have to be adapted. Hence, the price to pay for the overload is delay (which in turn is limited by the mobility of the channel).

\subsection{A proof sketch} The proof of the Theorem \ref{th:mainresult} will be long and technical. Let us therefore here first sketch its main steps.  In the first step  of the algorithm, we are investigating the value of
\begin{align*}
     \nu(\omega) = \sum_{i \in [t]} \norm{(\AAA^i_\omega)^*\bb^i}^2
\end{align*}
for different $(1,k_s)$-sparse supports $\omega$. Within each block, we then determine the $\omega$ which gives the highest value of $\nu(\omega)$. Clearly, we may just as well compare the average of those expressions over $i$, i.e.
\begin{align*}
\overline{\nu}(\omega) = \tfrac{1}{t} \sum_{i \in [t]} \norm{(\AAA^i_\omega)^*\bb^i}^2.
\end{align*}

Using that $\bb^i = \AA^i(\hh^i + \zz^i)$, we have 
\begin{align*}
    \overline{\nu}(\omega)= & \tfrac{1}{t}\sum_{i \in [t]} \norm{(\AAA^i_\omega)^* \AAA^i(\hh^i+\zz^i)}^2
\end{align*}
We will in the following prove that as soon as $t$ and $m$ are large enough,
\begin{align*}
    \sup_{\omega \ (1,k_s)\text{-sparse.}}\abs{\overline{\nu}(\omega) - \norm{\hh_\omega}^2}\leq \text{const.} \cdot  \tfrac{1}{k_u}\norm{\hh}^2
\end{align*}
From that, we will be able to deduce that 
\begin{align*}
    \min_{k \, : \, \hh_k\neq 0} \sup_{\substack{\omega \text{ $(1,k_s)$-sparse.} \\ \omega \text{ in block $k$}}} \overline{\nu}(\omega) \geq \max_{k \, : \, \hh_k = 0 } \sup_{\substack{\omega \text{ $(1,k_s)$-sparse.} \\ \omega \text{ in block $k$}}} \overline{\nu}(\omega),
\end{align*}
which means that all users will be correctly classified as active by our algorithm. An intuitive sketch of how we are going to prove the above is as follows.
\begin{itemize}
    \item Conditioned on the draws of $\AAA^i$ and $\zz^i$, each of the terms above are averages of independent data variables. This means that they should concentrate around their expected values 
    as soon as $t$ is reasonably large.
    \item Next, we move on to analyse those expected values. 
    These values are affected by two layers of randomness: the randomness of the noise $\zz^i$ and the randomness of the matrices $\AAA^i$. We will therefore first bound the deviation caused by the noise 
    with high probability. 
    \item Finally, we will investigate the expected values without noise. 
    We will show that it is close to $\norm{\hh_\omega}^2$ under an assumption of the \emph{coherence} of the $\AAA^i$
    \begin{align*}
    \sup_{k \neq \ell} \abs{\sprod{\aaa_k,\aaa_\ell}}.
\end{align*}
 which we then argue will hold with high probability. In particular, in the case of independent $\AA^i$, we will leverage the averaging over time to allow this bound to be slightly higher than in the constant case. This makes it possible to get rid of a term  $\log(n)^{-1}$ in the number of allowed users.
\end{itemize}

Our argumentation differs from standard compressed techniques in several ways.   In comparison to the general results on hierarchically sparse recovery \cite{cosipChapter}, our results do not rely directly on HiRIP (hierarchical restricted isometry) properties of the measurement operators. In particular, the only property of the $\AAA^i$ that we use is that it has a small \emph{mutual coherence} with high probability. Therefore, the results in this section apply to much more general $\AAA^i$ than randomly subsampled DFT matrices.

     Furthermore, we rely on concentration both due to the random nature of the measurements $\AAA^i$ \emph{and} the data $d_k^i$. 
  This has the technical consequence that terms of order $4$ in the noise vectors appear in the expressions we need to prove concentration for. As a result, we need to use results that to the best of our knowledge has not been utilized for compressed sensing before.

\subsection{Mutual coherence}
All of our proofs will rely on the coherence of the matrices $\AAA^i$ being bounded. If we could deterministically find a way to make the coherence low, we could possibly arrive at a guarantee involving 'less' randomness. This is however not possible while keeping the number of measurements $m$ under control. To explain this, recall that the mutual coherence is lower bounded by the so called  Welch bound \cite{welchOG}. It states that if $(\aaa_k)_{k \in [n]}$ is a set of normalized vectors in $\C^n$, the coherence fulfills
\begin{align}
    \sup_{k \neq \ell} \abs{\sprod{\aaa_k,\aaa_\ell}} \geq \sqrt{\frac{n-m}{m(n-1)}},
\end{align}
with equality if and only if $(\aaa_k)_{k\in [n]}$ is an \emph{equiangular tight frame}, i.e has a constant value for $\abs{\sprod{\aaa_k,\aaa_\ell}}$ for all $k\neq \ell$ and fulfills $\sum_{k \in [n]} \aaa_k\aaa_k^*= \tfrac{n}{m} \II_n$. 

Note that the Welch bound can be rewritten as 
\begin{align*}
    m \geq \tfrac{1}{\sup_{k \neq \ell}\abs{\sprod{\aaa_k,\aaa_\ell}}^2 + \tfrac{1}{n-1}}.
\end{align*}
That is, for large $n$, a condition of the form $\sup_{k \neq \ell}\abs{\sprod{\aaa_k,\aaa_\ell}} \leq
\frac{\tau}{\sqrt{k_uk_s^2}}$ necessarily implies that $m\geq k_uk_s^2$. 
As a detailed reading of the proof of the main theorem shows, this sample complexity  would be totally acceptable for our needs. However, it was shown in \cite{welchbound,haviv2017restricted} that a subsampled DFT matrix only achieves the Welch bound if $\calB$ is a difference set. A set $\calB$ is a difference set if the sequence $(i-j)_{i\neq j \in \calB}$ contains each value in $\set{1, \dots, n-1}$ an equal amount of times, say $\lambda$ times. Such sets however necessarily satisfy
\begin{align*}
    m(m-1)=\lambda(n-1),
\end{align*}
i.e. the number of elements $m \geq \sqrt{n}$. 
Hence, the DFT matrices achieving the Welch bound are not useable for our needs.

We can however show that if we subsample the DFT matrix randomly, it will have a mutual coherence smaller than $\frac{\tau}{\sqrt{k_uk_s^2}}$ with high probability already when $m\geqsim 
k_uk_s^2\log(n)$, which in turn will be enough to prove our main result.

\begin{theorem}(\cite[Corollary 12.14]{FouRau2013}). \label{th:coherencebound} The randomly subsampled DFT matrix $\AAA$ fulfills
\begin{align*}
    \mathbb{P}\left( \sup_{k \neq \ell \in J} \abs{\sprod{\aaa_k,\aaa_\ell}} > \tfrac{\tau}{\sqrt{k_uk_s^2}}\right) \leq 2n^2 \exp\left(-\tfrac{3m\tau^2}{k_uk_s^216}\right).
\end{align*}
\end{theorem}

Let us end by noticing that from a practical point of view, our analysis resting upon the coherence of the matrix rather than its RIP is beneficial. Whether a matrix obeys a coherence bound is immediate to check, whereas checking whether it obeys the RIP is NP-hard \cite{Bandeira2013}. Practically, this means that once a user has been distributed into a group $\calB^i_j$, he or she can check the coherence of the respective ${\AA^i}$. If the coherence is high, they will know that their recovered vector probably cannot be trusted.

\subsection{Auxiliary results}

The proof of our main recovery result is very technical, and rests on a collection of auxiliary results. Let us state, and prove some of them, here. The first result shows measure concentration with respect to the random draw of the data. 
\begin{lemma} \label{lem:tconc}
    Let $\AAA^i,\BB^i \in \C^{m,n}$, and  $\vv^i = (\vv^i_0, \dots, \vv^i_{n_s}) \in \C^n, i \in [t]$ be fixed, with $\vv^i_p \in \C^s$ and $n_s := \lceil \tfrac{n}{s}\rceil$. Let furthermore $(d^i_p)_{i \in [t], p \in [n_s]}$ be random variables, identically and independently distributed according to a distribution $d$ on the unit circle in the complex plane, with
    \begin{align*}
        \erw{d}=0.
    \end{align*}
    Letting $\AAA_\ell \in \C^{m,s}$ denote the $\ell$:th block of $\AAA$, $\AAA= [ \AAA_0, \dots, \AAA_{\lfloor \frac{n}{s}\rfloor}]$, define $\DD^i$ as the block diagonal matrix with $p$:th block $d^i_p \II_s$, and 
    \begin{align*}
        \Psi(\AAA,\BB,\vv) &= \tfrac{1}{t}\sum_{i \in [t]}\sum_{p \in [n_s]}\sprod{\AAA^i_p\vv_p^i , \BB^i_p \vv_p^i} \\
        \psi(\AAA,\BB,\vv) &= \sup_{i \in [t]} \sum_{ p\neq q} \abs{\sprod{\AAA_p^i\vv_p^i,\BB_q^i\vv_q^i}}^2  .
    \end{align*}
    Then,
    \begin{align*}
       &\mathbb{P}\left(\bigg\vert \tfrac{1}{t} \sum_{i \in [t]} \sprod{\AAA
       ^i \DD^i\vv^i,\BB^i\DD^i \vv^i} - \Psi(\AAA,\BB,\vv) \bigg\vert > \rho \right) \\
        &\quad \leq 2\exp\left( -\kappa t \min\left(\tfrac{\rho^2}{\psi(\AAA,\BB,\vv)},\tfrac{\rho}{\psi(\AAA,\BB,\vv)^{\sfrac{1}{2}}}\right)\right),
    \end{align*}
    where $\kappa$ is a universal constant.
    
\end{lemma}
\begin{proof}[Proof of Lemma \ref{lem:tconc}]
    We have
    \begin{align*}
        &\tfrac{1}{t} \sum_{i \in [t]} \sprod{\AAA^i \DD^i \vv^i,\BB^i\DD^i\vv^i} \\
        &\quad = \tfrac{1}{t} \sum_{i \in [t]}\sum_{p,q \in [r]} \overline{d_p^i}d_q^i\sprod{\AAA_p^i\vv_p^i, \BB_q^i \vv_q^i}\\
        &\quad = \Psi(\AAA,\BB,\vv)+ \tfrac{1}{t} \sum_{i \in [t]}\sum_{p\neq q \in [r]} \sprod{\AAA_p^i\vv_p^i, \BB_q^i\vv_q^i} \overline{d_p^i}d_q^i,
    \end{align*}
    since for all $i$ and $p$, $\overline{d_i(p)}d_i(p)=\abs{d_i(p)}^2=1$.  By defining a matrix $\MM \in \C^{n_st,n_st}$ through its blocks
    \begin{align*}
        \MM_{(i,p),(j,q)} = \begin{cases} t^{-1}\sprod{\AAA_p^i\vv_p^i, \BB_q^i\vv_q^i} &\text{ if } i=j, p \neq q \\
        0 &\text{ else,}
        \end{cases}
    \end{align*}
    and $\DDelta \in \C^{n_st}$ through $\DDelta_{j,q} = d_q^j$, the random variable we are trying to bound is equal to $\sprod{\DDelta,\MM\DDelta}$. For this variable, we may invoke the \emph{Hanson-Wright}-inequality (see e.g \cite{rudelson2013hanson}), which states that
    \begin{align*}
        &\prb{ \abs{\sprod{\DDelta,\MM\DDelta} - \erw{\sprod{\DDelta,\MM\DDelta}}} > \rho} \\ & \qquad \leq 2\exp\left( - c \min \bigg(\frac{\rho^2}{\norm{\MM}_F^2}, \frac{\rho}{\norm{\MM}}\bigg)\right)
    \end{align*}
    In this case, the expected value is zero and
    \begin{align*}
         \norm{\MM}_F^2 &= \tfrac{1}{t^2} \sum_{i \in [t]} \sum_{ p\neq q} \abs{\sprod{\AAA_p^i\vv_p^i,\BB_q^i\vv_q^i}}^2 \\
         & \leq t^{-1}\psi(\AAA,\BB,\vv) \\
         \norm{\MM} &= \sup_{i \in [t],\norm{\uu},\norm{\ww}\leq 1}\tfrac{1}{t} \abs{\sum_{p \neq q} \sprod{\AAA_p^i\vv_p^i, \BB_q^i\vv_q^i}\overline{u_p}w_q} \\
         &\\
         &\quad  \leq t^{-1}\sup_{i \in [t]}\left(\sum_{p \neq q} \abs{\sprod{\AAA_p^i\vv_p^i,\BB_q^i\vv_q^i}}^2\right)^{\sfrac{1}{2}} \\
         &= t^{-1} \psi(\AAA,\BB,\hh)^{\sfrac{1}{2}}
    \end{align*}
   which gives the claim.
\end{proof}

The above theorem suggests that we now need to analyse the expressions $\Psi(\AAA,\BB,\vv)$ and $\psi(\AAA,\BB,\vv)$ for $\vv=\hh+\zz$ and \begin{equation}
    \BB^i_\omega = \AAA^i_\omega (\AAA^i_\omega)^*\AAA^i, i\in [t]
    \label{eqn:B_omega}
\end{equation}
for different values of $\omega$. The smaller we can get them, the tighter the above bound will be. The first step is to analyse the effect of the noise, i.e. to bound the deviations $\abs{\Psi(\AAA,\BB,\hh+\zz)-\Psi(\AAA,\BB,\hh)}$ and $\abs{\psi(\AAA,\BB,\hh+\zz)-\psi(\AAA,\BB,\hh)}$. Here, we mainly use the Gaussianity of $\zz$. We arrive at the following results.

\begin{lemma} \label{lem:Psi_variation}
    For an arbitrary $(1,k_s)$-sparse support $\omega$,
    define
    \begin{align*}
        \mathfrak{h}^2 &= \norm{\hh_\omega}^2 + 2\tfrac{1}{\sqrt{k_u}}\norm{\hh}\norm{\hh_\omega} + \tfrac{1}{k_u}\norm{\hh}^2 
    \end{align*}
    Under the assumption that the coherences of all $\AAA^i$ are smaller than $\tfrac{\tau}{\sqrt{k_u k_s^2}}$, we have
    \begin{align*}
         &\abs{ \Psi(\AAA,\BB_\omega,\hh+\zz) - \Psi(\AAA,\BB_\omega,\hh)} \leq    \mathfrak{h}\tfrac{2\sigma m^{\sfrac{1}{2}}}{n^{\sfrac{1}{2}}} + \tfrac{\sigma^2m}{n}
    \end{align*}
    with a failure probability smaller than
    \begin{align*}
         4t^{1-\kappa n} + 2 \exp( - \tfrac{\kappa t}{\log(t)^2(1 + \tfrac{\tau^2}{k_uk_s})\max(1,\tau)^2}).
    \end{align*}
\end{lemma}

\begin{lemma} \label{lem:psivariation}
For an arbitrary $(1,k_s)$-sparse support $\omega$,
define
\begin{align*}
    \mathds{h}^2 =& \norm{\hh_\omega}^2 + \tfrac{1+\tau^2}{k_u} \norm{\hh}^2, \qquad  
    \otau^2 = \tfrac{\tau^2}{k_u}, \qquad \osigma^2 = \tfrac{\sigma^2m}{n} \\
    \Delta  =& \log(rt\nicebinom) \big(\psi(\AAA,\BB,\hh)^{\sfrac{1}{2}}(\osigma \mathds{h}\otau + \osigma^2\otau(1+\otau) \\
      & \qquad + \osigma^2(\sigma^2(1+\otau)^2\otau^2 + \osigma \otau(1+\otau)\mathds{h} + \otau^2\mathds{h}^2 \big)
\end{align*}
Under the assumption that the coherences of all $\AAA^i$ are smaller than $\tfrac{\tau}{\sqrt{k_u k_s^2}}$, $\psi(\AAA,\BB_\omega,\hh+\zz)$  is smaller than
    \begin{align*}
           \psi(\AAA,\BB_\omega,\hh) +C( \osigma^2 \otau^2\mathds{h}^2 + \osigma^4(1+\otau)^2\otau^2 + \Delta) \\
    \end{align*}
    with a failure probability smaller $t^{1-\kappa}r^{-\kappa}\binom{s}{k_s}^{-\kappa}$, where $\kappa$ is a constant dependent on the value of the constant $C$.
\end{lemma}  

The proofs are long and technical and are therefore postponed to Section \ref{sec:devlem} in the appendix.

The next step of the argument is to bound the 'undeviated' expressions $\Psi(\AAA,\BB_\omega,\hh)$ and $\psi(\AAA,\BB_\omega, \hh)$. We arrive at the following two lemmas. The proofs are, albeit conceptually easy, quite long and technical, whence we postpone them to Section \ref{sec:devlem} of the appendix. The main idea is to utilize the low mutual coherence - and in particular only argues probabilistically in the case of independent $\AAA^i$. 
\begin{lemma} \label{lem:Psiraw}
Assume that the coherences of  $\AAA^i$ all are bounded by $\tfrac{\tau}{\sqrt{k_uk_s^2}}$. Then
\begin{align*}
    \abs{ \Psi(\AAA,\BB_\omega,\hh)- \norm{\hh_\omega}^2 } &\leq \tfrac{\tau^2\norm{\hh}^2}{k_u} + \tfrac{2\tau\norm{\hh_\omega}\norm{\hh}}{\sqrt{k_u}} 
\end{align*}
Additionally, if the $\AAA^i$ are independent, we have for $\gamma>0$
\begin{align*}
    \abs{ \Psi(\AAA,\BB_\omega,\hh)- \norm{\hh_\omega}^2 } &\leq \tfrac{\gamma^2\norm{\hh}^2}{k_u} + \tfrac{\gamma^2\norm{\hh_\omega}\norm{\hh}}{\sqrt{k_u}}
\end{align*}
with a failure probability smaller than
    $4\exp\left(-\tfrac{t}{\max(\sfrac{\tau^2}{\gamma^2},\sfrac{\tau^4}{\gamma^4})}\right)$.
\end{lemma}

\begin{remark}
We emphasize at this point that Lemma \ref{lem:Psiraw} is quite essential for the main result in the next section. In fact, we see that for the independent case, the $\tau$ parameter  or equivalently the coherence, can be large with some small probability{. We can use the $\gamma$-parameter to compensate for this, and make the} term $\Psi(\AAA,\BB_\omega,\hh)$ still rightly concentrate around its expected value $\norm{\hh_\omega}^2$. This will be pivotal for the overload situation!
\end{remark}
\begin{lemma} \label{lem:psiraw}
    Under the assumption that all coherences of the $\AAA^i$ are bounded by $\tfrac{\tau}{\sqrt{k_uk_s^2}}$, we have
    \begin{align*}
        \psi(\AAA,\BB_\omega,\hh) \leq \tfrac{3\tau^2}{k_u}\norm{\hh_\omega}^2\norm{\hh}^2 + \tfrac{3\tau^4}{k_u^2}\norm{\hh}^4.
    \end{align*}
\end{lemma}


\subsection{Proof of the main result}
We can now prove Theorem \ref{th:mainresult}

\begin{proof}[Proof of Theorem \ref{th:mainresult}]
    Let us fix the sub-channel, set $\beta= \tfrac{32\kappa}{3C_0}$, and fix the number of measurements in the sub-channel to 
    \begin{align*}
        m=\beta \log(n) \tfrac{n}{u}.
    \end{align*}
    Then, Theorem \ref{theo:sparsity} shows that the  effective vector $\hh$ is $(2m\cdot \tfrac{u}{n},k_s)$-sparse with failure probability smaller than
    \begin{align*}
        \exp(-\tfrac{3mu}{4n}) + \tfrac{4m^2u^2}{n^2r} &= \exp(-\tfrac{3\beta\log(n)}{4}) + \tfrac{4\beta^2\log(n)^2}{r}\\
        &\leq n^{-\sfrac{3\beta}{4}} + \sfrac{\epsilon}{2},
    \end{align*}
    where we used our assumption on the size of $r$ in the final step.
 We may without loss of generality assume that $k_u\geq 1$ -- if not, there are no users to detect in the channel.
    
     Theorem \ref{th:coherencebound} shows that the coherence of each $\AAA^i$ is smaller than $\tfrac{\tau^2}{\sqrt{k_uk_s^2}}$ with a failure probability smaller than
    \begin{align*}
        n^2\exp\left(-\tfrac{3m\tau^2}{32k_uk_s^2}\right) & = n^2\exp\left(-\tfrac{3\beta n\tau^2}{32 u k_s^2}\right) 
    \end{align*}
    Now, in the case of a single measurement operator $\AAA$, our corresponding assumption on $u$ reads $\tfrac{n}{u} \geq C_o k_s^2 \log(n)$ for some constant $C_o$. Setting $\tau$ equal to $1$, the above is hence smaller than 
    \begin{align*}
        n^2\exp\left(-\tfrac{3\beta C_olog(n)}{16}\right) & = n^{2-\tfrac{3\beta C_o}{16}} = n^{2-\kappa} 
    \end{align*}
    where we in the third step used that $\beta = \tfrac{16\kappa}{3C_o}$. 
    
    In the case of independent $\AAA^i$, the assumption on the number of measurements instead implies $\tfrac{n}{u} \geq C_o k_s^2$ again for some $C_o$. Setting $\tau= \log(tn)^{\sfrac{1}{2}}$ and applying a union bound  over $i$ implies that all $\AAA^i$ have a coherence smaller than $\tfrac{\log(tn)^{\sfrac{1}{2}}}{\sqrt{k_uk_s^2}}$ with a probability
    \begin{align*}
        tn^2\exp\left(-\tfrac{3\beta C_o\log(tn)}{32}\right) & \leq tn^2\cdot tn^{-\tfrac{3\beta C_o}{32}} \\
        &= t^{1-\kappa}n^{2-\kappa}  \leq n^{2-\kappa},
    \end{align*}
    where we in the third step again used $\beta=\tfrac{16\kappa}{3C_o}$, and $t\geq 1$ in the final step. In the following, let us for convenience write $\tau$ for $1$ in the case of a single $A$ and $\log(tn)$ for the case of independent $\AAA^i$.
    
    Our power assumption of the active blocks implies
    $\tfrac{3}{4}k_u k_s \leq\norm{\hh}^2 \leq \tfrac{5}{4}k_uk_s$, so that
        \begin{align*}
            \tfrac{3}{5k_u} \norm{\hh}^2 \leq \norm{\hh_k'}^2 \leq \tfrac{5}{3k_u} \norm{\hh}^2,
    \end{align*}
    and consequently
    \begin{align}  \label{eq:hbound}
        \norm{\hh_\omega}^2 \leqsim \tfrac{\norm{\hh}^2}{k_u}.
    \end{align}
    Also, using the notation of Lemma \ref{lem:psivariation}, our assumption on the noise level implies that 
    \begin{align} \label{eq:varbound}
        \osigma^2 &= \tfrac{\sigma^2 m}{n} = \tfrac{\sigma^2 m }{n} \cdot \tfrac{2mu}{nk_u} = \nonumber \tfrac{2\beta^2\sigma^2un^2\log(n)^2}{n^2u^2} \\
                    &= \tfrac{\beta^2\sigma^2\log(n)^2}{u} =   \tfrac{\norm{\hh}^2}{k_u}.
    \end{align}
        This has a number of consequences. First, with Lemma \ref{lem:psiraw}, we may estimate
    \begin{align*}
        \psi(\AA,\BB,\hh) \leqsim (\tau^2+\tau^4)\tfrac{\norm{\hh}^4}{k_u^2} \leq \tau(1+\tau)^3\tfrac{\norm{\hh}^4}{k_u^2}
    \end{align*}
    Also, again using the notation of Lemma \ref{lem:psivariation}, we can estimate
    \begin{align*}
        \mathds{h}^2 &\leqsim (1+\tau^2)\tfrac{\norm{\hh}^2}{k_u} \\
        \Delta & \leqsim (\log(r)^2 + k_s\log(s)^2)\tfrac{\norm{\hh}^4}{k_u^2}\tau(1+\tau)^3,
    \end{align*}
    where we estimated $\otau \leq \tau$. 
    All in all, we conclude that Lemma \ref{lem:psivariation} shows that with a failure probability smaller than $t^{1-\kappa} r^{-\kappa}\nicebinom^{-\kappa}$, we have
    \begin{align*}
        \psi(\AA,\BB,\hh+\zz) \leqsim C&\tfrac{\norm{\hh}^4}{k_u^2} \tau(1+\tau)^3 \cdot \log(tr\nicebinom)
    \end{align*}
    By a union bound, the above is true for all $r\nicebinom$ $(1,k_s)$-sparse supports $\omega$ with a failure probability smaller than $(tr\nicebinom)^{1-\kappa}$.
    
     In particular, the above implies that for every $\gamma~>~0$
    \begin{align} \label{eq:psicontrol}
\frac{\sfrac{\norm{\hh}^4\gamma^2}{k_u^2}}{\psi(\AA,\BB_\omega,\hh+\zz)} \geqsim \tfrac{1}{\gamma^2\tau(1+\tau)^3\log(tr\nicebinom)}.
\end{align}
    
    Next, we move on to the expressions involving $\Psi$.  Our bounds \eqref{eq:hbound} and \eqref{eq:varbound} imply
    \begin{align*}
        \mathfrak{h}^2 &\leqsim \tfrac{\norm{\hh}^2}{k_u} \\
        2\mathfrak{h}\osigma + \osigma^2 &\leqsim \tfrac{\norm{\hh}^2}{k_u},
    \end{align*}
    where we used the notation defined in Lemma \ref{lem:Psi_variation}. That lemma, together with a union bound,  hence implies that
    \begin{align} \label{eq:Psirawbound}
        &\abs{ \Psi(\AA,\BB_\omega,\hh+\zz) - \Psi(\AA,\BB_\omega,\hh) } \leqsim \tfrac{\norm{\hh}^2}{k_u},
    \end{align}
    for every support $\omega$ with a failure probability smaller than 
    \begin{align*}
        r \binom{s}{k_s}& \left(4t^{1-\kappa n} +\exp\left(-\frac{\kappa t}{\log(t)^2(1 + \tau^2)\tau^2}\right)\right).
    \end{align*}    
    Due to the assumptions on the $t$-parameter in the respective cases, we have
    \begin{align*}
        r \binom{s}{k_s}\exp\left(-\frac{\kappa t}{\log(t)^2(1 + \tau^2)\tau^2}\right) \leq \epsilon
    \end{align*}
 in both cases.

Lemma \ref{lem:Psiraw} further implies that 
    \begin{align*}
         \abs{ \Psi(\AA,\BB_\omega,\hh)- \norm{\hh_\omega}^2 } &\leq \gamma \tfrac{\norm{\hh}^2}{k_u},
    \end{align*}
    deterministically in the case of equal $\AA^i$, and  with  a failure probability smaller than $4\exp\left(-\gamma\tfrac{t}{\tau^4})\right)\leq \epsilon$ in the independent case. Importantly, we may choose the size of $\gamma$ adaptively by adjusting the values of the implicit constant in the bound on $t$.

Now, set $\rho =\gamma \tfrac{\norm{\hh}^2}{k_u}$. Applying Lemma \ref{lem:tconc}, in particular utilizing the bound \eqref{eq:Psirawbound} implies that we have
\begin{align}
    \bigg\vert \tfrac{1}{t} \sum_{i \in [t]} \sprod{\AA \DD_i (\hh^i+\zz^i),\BB_\omega \DD_i (\hh^i+\zz^i)} - \norm{\hh_\omega}^2\bigg\vert \nonumber \\
    \leqsim \gamma\tfrac{\norm{\hh}^2}{k_u} \label{eq:conc}
\end{align}
for all $\omega$ with a failure probability smaller than 
\begin{align*}
    r\binom{s}{k_s} \exp\left(- \kappa \tfrac{t}{\gamma^2\tau(1+\tau)^3\log(tr\nicebinom)}\right),
\end{align*}
which is smaller than $\epsilon$ in both cases.

Now it is clear that as long as the implicit constant in \eqref{eq:conc} is smaller than $\gamma=\tfrac{1}{5}$, that equation will imply that the best $k_s$-approximations $\eta_j$ of the active blocks $\omega$ will obey
\begin{align*}
    \norm{\eeta_j^*} \geq \norm{\hh_\omega}^2 -\tfrac{\norm{\hh}^2}{5k_u} \geq\norm{\hh_\omega}^2 -\tfrac{\norm{\hh}^2}{5k_u} \geq \tfrac{2}{5k_u}\norm{\hh}^2,
\end{align*}
and if the block is inactive (so that $\norm{\hh_\omega}^2=0$),
\begin{align*}
    \norm{\eeta_j^*}  \leq \frac{\norm{\hh}^2}{5k_u}.
\end{align*}
This shows that if we choose the threshold $\theta= \tfrac{3\norm{\hh}^2}{10k_u}$, the $k_u$ blocks that will be chosen by the algorithm are exactly the ones that are active. The proof is finished.
\end{proof}


\section{Recovery analysis}
Above, we provided conditions for all users within a sub-channel to be correctly classified as active. In this section we derive some guarantees for the actual data detection. To do so, we assume for simplicity that the system is \emph{not overloaded}, {i.e. that $u\leq n$}. This will allow us to bypass information-theoretic treatment of multiuser detection involving successive interference cancellation, pseudo-inverses etc. In other words, it is merely a stability result for the under-determined systems to be solved.
We need to make statements about the solutions of the restricted least squares problems
\begin{align*}
    \min_{\substack{\hh^i \\ \supp \hh_i \sse \Omega}} \norm{\AA^i \hh^i - \bb^i}^2
\end{align*}
It will turn out that the analysis of this problem is very simple as long as $n$ is prime and $\Omega$ is exactly equal to the support of $\hh$. Note that this does not automatically follow from the previous section -- there, we only show that the blocks are correctly classified. It can very much be that the best $k_s$-sparse approximations within the blocks are not equal to the $\hh_i$ of the original vectors, in particular considering noise effects.

By a simple trick, it is however not hard to give conditions for when $\Omega = \supp \hh$ with high probability. Notice that when applying Theorem \ref{th:mainresult}, there is nothing stopping us from regarding $h$ as a vector of $rs$ blocks of length $1$, out of which $k_uk_s$ are active and $1$-sparse (instead of $r$ blocks of length $s$, out of which $k_u$ are active and $k_s$-sparse). From this point of view however, all assumptions of the theorem are not met -- what is missing is the power assumption. But if that assumption is made, which exactly corresponds to all non-zero entries of all $\hh_k'$ having essentially constant magnitudes, i.e.
\begin{align} \label{eq:constmag}
    \abs{\hh_k'(\ell)} \sim 1,
\end{align}
the theorem shows that under the same conditions as before, all \emph{entries} $\hh_k(\ell)$ of the vector are correctly classified -- i.e., the support is exactly recovered. Let us record this as a corollary.
\begin{corollary} \label{cor:supp}
    In addition to the assumptions of Theorem \ref{th:mainresult}, assume the fine-grained power control \eqref{eq:constmag}, and that the constant $C_o<k_s^{-1}$. Then, our algorithm exactly recovers the support of $\hh$.
\end{corollary}

Once the support of $\hh$ has been detected, we may recover all $\hh^i$ by solving the least squares problem
\begin{align}
    \min_{\supp \xx \sse \supp \hh} \norm{\AA^i\xx-\bb^i} \label{eq:ls}
\end{align}
In order to show that this problem succeeds at approximately recovering the $\hh^i$, we need the following theorem. 
\begin{theorem} [Reformulation of \cite{tao2005uncertainty}, Theorem 1.1] \label{th:Fsubmatrix}
    Suppose that $n$ is prime and {$\mathbf{F}$} is {the $n\times n$} DFT matrix. Then, every square sub-matrix of {$\mathbf{F}$} is invertible. In particular, any submatrix formed by concatenating at least as many columns as the number of rows is injective.
\end{theorem}
\begin{remark}
    The assumption that $n$ is prime is a technical one, but really  needed for Theorem \ref{th:Fsubmatrix} to hold. To see why, assume that $n=pq$ for some integers $p,q$. Define $I=p\cdot [q]$ and let $\vv$ be the indicator of $I$, i.e. the vector equal to $1$ on $I$ and else zero. Then, we have for $k$ with $e^{-\sfrac{2\pi \iota k }{q}} \neq 1$, i.e. $k \notin q \cdot [p]$
    \begin{align*}
        (\mathbf{F}\vv)_k &= \sum_{\ell \in I} e^{-\sfrac{2\pi \iota k \ell}{n}} = \sum_{j \in [q]} e^{-\sfrac{2\pi \iota k jp}{pq}}\\
        &= \frac{1-e^{-\sfrac{2\pi \iota  q}{q}}}{1-e^{-\sfrac{2\pi \iota k }{q}}} = 0.  
    \end{align*}
    Hence, for any $J\sse [n]$ with $\abs{J}=\abs{I}$ and $J \cap q \cdot [p]= \emptyset$, the square submatrix $(F_{ji})_{j \in J,i \in I}$ has a non-trivial kernel and hence is non-invertible.

    Note however that since we can choose $n$  ourselves, the primality of $n$ is actually not a restriction.
\end{remark}

We may now prove that under the same conditions as in our main theorem, and that $n$ is prime, all $\hh^i$ will be succesfully recovered.  It is in fact at this point relatively simple to prove estimates for the data recovery.
\begin{lemma}
    Under the same assumptions as the main theorem, with the additional condition that $n$ is prime, all solutions $\hh^i_*$ of \eqref{eq:ls} are given by
    \begin{align*}
        \hh_*^i = \hh^i + \zz^i.
    \end{align*}
    with failure probability smaller than $n^{-\sfrac{3\beta}{4}} + \epsilon.$
\end{lemma}
\begin{proof}
    Let us for notational simplicity set  $\Omega=\supp \hh$ and $\aalpha = (\AA^i)_\Omega$. It is clear that the solution of \eqref{eq:ls}, as soon as $\aalpha^*\aalpha$ is invertible, is given by
    \begin{align*}
        &(\aalpha^*\aalpha)^{-1}\aalpha^* \bb \\
        & \quad =  (\aalpha^*\aalpha)^{-1}\aalpha^*\aalpha(\hh+\zz) = \hh+\zz.
    \end{align*}
    However, by Theorem \ref{th:Fsubmatrix} and the primality assumption on $n$, $\aalpha^*\aalpha$ is invertible, as soon as $m \geqsim \abs{\Omega}=k_uk_s$. So let us argue that $m$ is of that size with high probability. Our model implies that with a failure probability less than the one given in the main theorem, $k_u \leq 2 m \tfrac{u}{n}$. This proves
    \begin{align*}
        \tfrac{m}{k_uk_s} \geq \tfrac{m}{2m \tfrac{u}{n} k_s} = \tfrac{n}{2uk_s}\geqsim 1
    \end{align*}
    due to the assumption $u \leqsim \tfrac{n}{k_s^2}$. The theorem has been proven.
\end{proof}

The above lemma tells us that the recovered vector actually is equal to the ground truth contaminated by a Gaussian vector. This has the following consequence for the recovery of the data.
\begin{theorem}
    Consider the data $(d^i(\ell))_{\ell \in \omega}$, with $\abs{\omega}=k_s$. Under the additional assumption that $\sigma\left(\tfrac{\log(n)}{u}\right)^{\sfrac{1}{2}} \leqsim \inf_\ell \abs{h(\ell)}$, we have
    \begin{align*}
      \sup_{i\in [t], \ell \in \omega} \abs{d_*^i(\ell)-d^i(\ell)} \leqsim \sigma \left(\tfrac{\log(t)\log(n)}{u}\right)^{\sfrac{1}{2}}\abs{h(\ell)}^{-1}
    \end{align*}
with a failure probability smaller than $n^{-\sfrac{3\beta}{4}}+\epsilon + 2k_s(e^{-n} + t^{1-n})$.
\end{theorem}
\begin{proof}
    By the previous lemma, with failure probability smaller than $n^{-\sfrac{3\beta}{4}} + \epsilon$, $\hh^i_* = \hh^i + \zz^i$. Therefore, for each $\ell \in \omega$,
    \begin{align*}
        \dd^i_*(\ell)  =\frac{\hh_*^i(\ell)}{\hh_*^0(\ell)} = \frac{\hh^i(\ell) + \zz^i(\ell) }{\hh^0(\ell) + \zz^0(\ell)} = \frac{\dd^i(\ell)\hh(\ell) + \zz^i(\ell) }{\hh(\ell) + \zz^0(\ell)}.
    \end{align*}
    Since the $\zz^i$ are Gaussians and independent of $\dd^i$, the  above variable  has the same distribution as $\delta^i(\ell)= \dd^i(\ell) \tfrac{\hh(\ell) + \zz^i(\ell) }{\hh(\ell) + \zz^0(\ell)}$. Now, since the $\dd^i(\ell)$ have modulus $1$
    \begin{align*}
       \abs{\delta^i(\ell)- \dd^i(\ell)} &= \abs{\dd^i(\ell) \left(\tfrac{\hh(\ell) + \zz^i(\ell) }{\hh(\ell) + \zz^0(\ell)} - 1\right)} \\
       &\leq  \frac{\abs{\zz^i(\ell) -\zz^0(\ell)} }{\abs{\hh(\ell) + \zz^0(\ell)}}.
    \end{align*}
    By a simple union bound and utilizing that the noise is Gaussian, we  have $\abs{\hh(\ell) + \zz^0(\ell)}\geq \abs{\hh(\ell)} -\tfrac{\sigma m^{\sfrac{1}{2}}}{n^{\sfrac{1}{2}}}$  for all $\ell$ with a probability bigger than $2k_se^{-n}$. Furthermore, the variables $\zz^i(\ell) -\zz^0(\ell)$, $\ell \in \supp h, i \in [t]$, are Gaussians with variance $\tfrac{2\sigma^2m}{n^2}$, and hence, they are all smaller than $\tfrac{\sqrt{2}\sigma m^{\sfrac{1}{2}}\log(t)^{\sfrac{1}{2}}}{n^{\sfrac{1}{2}}}$ with a probability smaller than $2tk_se^{-n\log(t)}=2k_st^{1-n}$. 
    Also note that by the choice of $m$ described in Theorem \ref{th:mainresult}, we have 
    \begin{align*}
        \tfrac{m}{n} = \beta\log(n)\tfrac{1}{u} \leqsim \inf_\ell \abs{\hh(\ell)}^2 \sigma^{-2}
    \end{align*}
    for some constant $\beta$. Hence, assuming that the implicit constant above is chosen large enough, the two estimates above imply that
    \begin{align*}
        \frac{\abs{\zz^i(\ell) -\zz^0(\ell)} }{\abs{\hh(\ell) + \zz^0(\ell)}} &\leq \frac{\tfrac{\sqrt{2}\sigma m^{\sfrac{1}{2}}\log(t)^{\sfrac{1}{2}}}{n^{\sfrac{1}{2}}}}{\abs{\hh(\ell)} -\tfrac{\sigma m^{\sfrac{1}{2}}}{n^{\sfrac{1}{2}}}} \\
        &\leqsim  \sigma \left(\tfrac{\log(t)\log(n)}{u}\right)^{\sfrac{1}{2}}\abs{h(\ell)}^{-1},
    \end{align*}
  which was the claim.
\end{proof}

Let's give an interpretation of the last theorem. The additional assumption we're making is an assumption on the amplitude on the entries of the vectors $\hh$, in particular its relation to the noise level. This is natural --  if the entries are too small, they will drown in the noise. Similarly, the estimate we prove is also meaningful only when $\sigma \abs{\hh(\ell)}^{-1}$ is small enough. Note that the meaning of 'small enough' gets \emph{more relaxed} as the number of users increase -- in particular, if we choose $u \sim \tfrac{n}{k_s^2\log(rs)}$, which Theorem \ref{th:mainresult} tells us is possible, the maximal error will be very small unless $t$ is exponentially large in $\tfrac{n}{\log(n)^2}$. This is also natural -- when we send very long sequences of independent messages, the probability that at least one of them fails of course increases.

\section{Numerical Experiments}
We have carried out two numerical experiments: In the first experiment we let $n$ scale up and find the number of supported users given a predefined fixed probability of collisions in each of the sub-channels under noise. In the second experiment, we let the number of users $u$ scale up while $n$ is kept fixed, and count the correctly detected users under noise. It is important to emphasize that we have defined \emph{Signal-to-Noise Ratio} SNR as 
\begin{align*}¨
\mathrm{SNR}=10\log_{10}(\tfrac{1}{\sigma^2}) \mathrm{ dB}.
\end{align*}
Hence, the true physical noise in the system is 
\begin{align*}
\mathrm{SNR_{true}}=\mathrm{SNR}-10\log_{10}(n/m) \mathrm{ dB}
\end{align*}
(see our proxy measurement model) . By this definition, e.g, an SNR=20 dB corresponds to 13.0 dB for the experiments depicted in Figure 3-5 with $n=2048$ and $m=256$, which is a lot less!

Parallel sub-channels are
created by randomly partitioning the $n$-dimensional image space of a DFT matrix into blocks
of length $m$, leading to $c=n/m$ sub-channels. For simplicity, we set the number of available pilot resources to $r=n/s$. Consequently, for  each sub-channel
$j=1,\ldots,c$, a vector $\xx^{j}\in\mathbb{C}^{n}$ is divided into $r$ blocks, 
each of length $s$. Hence, we use the full column space of the DFT and do not change the measurement pattern for simplicity. If a user chooses a pilot supported in DFT domain on block $\calB^k$,
$k\in\lbrack r]$, block $\xx_{k}$ is filled with the $k$th user's $k_{s}$-sparse
signature. 
Each user is also
encoding data into diagonal matrices $\DD_{i}$, $i=2,\ldots,t$ containing
entries of modulus 1. Hence, at the access point, data blocks $\yy_{i}^{j}=\AA_{j}\DD_{i}%
\xx^{j}\in\mathbb{C}^{m}$ for $j=1,\ldots,c,i=1,\ldots,t$ are received, forming
the observation $\yy\in\mathbb{C}^{c\times m\times t}$. Here, $\AA_{j}$ is a
matrix consisting of $m$ rows of a $n\times n$ DFT matrix, corresponding to
the sub-carriers allocated to sub-channel $j$.
User detection is performed by one step of HiIHT \cite{Wunder2019_TWC,blumensath2009iterative}, a slight variant of the HiHTP \cite{Roth2020_TSP}.

In the first experiment we investigate the scaling of the system under ideal conditions. 
We fix the number of sub-carriers bundled together to form one pilot resource to $s=8$, resulting in $r=n/8$ available resources per sub-channel.
Since each user chooses their sub-channel and resource (i.e. the pilot sequence) within that sub-channel uniformly at random, on average each sub-channel will be filled with $\bar{k}_u = u/c$ users. 
The number $\bar{k}_u$ of users per sub-channel is chosen such that the
probability of two or more users trying to use the same pilot sequence is
below a preset probability $0<p_{u}<1$, i.e. the largest $\bar{k}_u\in\mathbb{N}$ such
that
\begin{equation}
\prod\limits_{i=1}^{\bar{k}_u}\left(  1-\frac{i}{r}\right)  \geq 1-p_{u}%
.\label{eq:select_ku}%
\end{equation}
The left hand side of inequality \eqref{eq:select_ku} is the probability that
each of $k$ indices out of $[r]$ selected uniformly at random are unique. 
 We set the number of sub-channels as {$c=n/m$ with} $m=2^{\lfloor\log_{2}(\bar{k}_{u}\cdot k_s)\rfloor}$ and adjust $t$ such that in noise-free simulations the user detection with HiIHT has a detection rate $\geq1-p_{md}$, where $p_{md}$ denotes the probability of misdetection. Setting $p_{md}=0.1$, the average collision probability $p_u=0.1$ and the in-block sparsity $k_s=4$ resulted in $t=100$. 
With these choices, on average the total number of supported users is given
by 
\begin{align}
    u = (1-p_{u})\cdot \bar{k}_u\cdot c\cdot (1-p_{md}). \label{eq:nbrOfUsers}
\end{align}
The number of
supported users in this setting for $n=2^{10},\ldots,2^{13}$ can be observed
in Figure \ref{fig:supported_users}. With an SNR $\geq-10$dB (note again that this is system SNR, which is lower than the true physical SNR!) the system performance is able to recover all users as in the noise free case, whereas the recovery breaks down for lower SNRs. To draw a comparison with other random access schemes (e.g. slotted ALOHA), we assume that there also the premise is to minimise collisions as done in our scheme in order to transmit the same amount of data in the same time as our scheme. We do not assume that the same pilot design is used though and hence let the users choose out of $n$ instead of $r$ resources. The scaling of randomly selecting $k$ out of $n$ resources under a constraint on the probability of collision is shown in Figure \ref{fig:collisions2}. By our sub-channeling approach we can serve many more user without increasing the collision probability and at the same time stabilise the user detection by making use of the full data transmission time $t$.
Comparing the orange curve in Figure \ref{fig:supported_users},which following our discussion, corresponds to the number of users that can be served by slotted ALOHA, with the curve for our method for an SNR$\geq -10$ dB in gives roughly \emph{a 20 fold increase in user capacity for the tested system dimensions $n$}.

\begin{figure}[ptb]
    \centering
    \includegraphics[width=.49\textwidth]{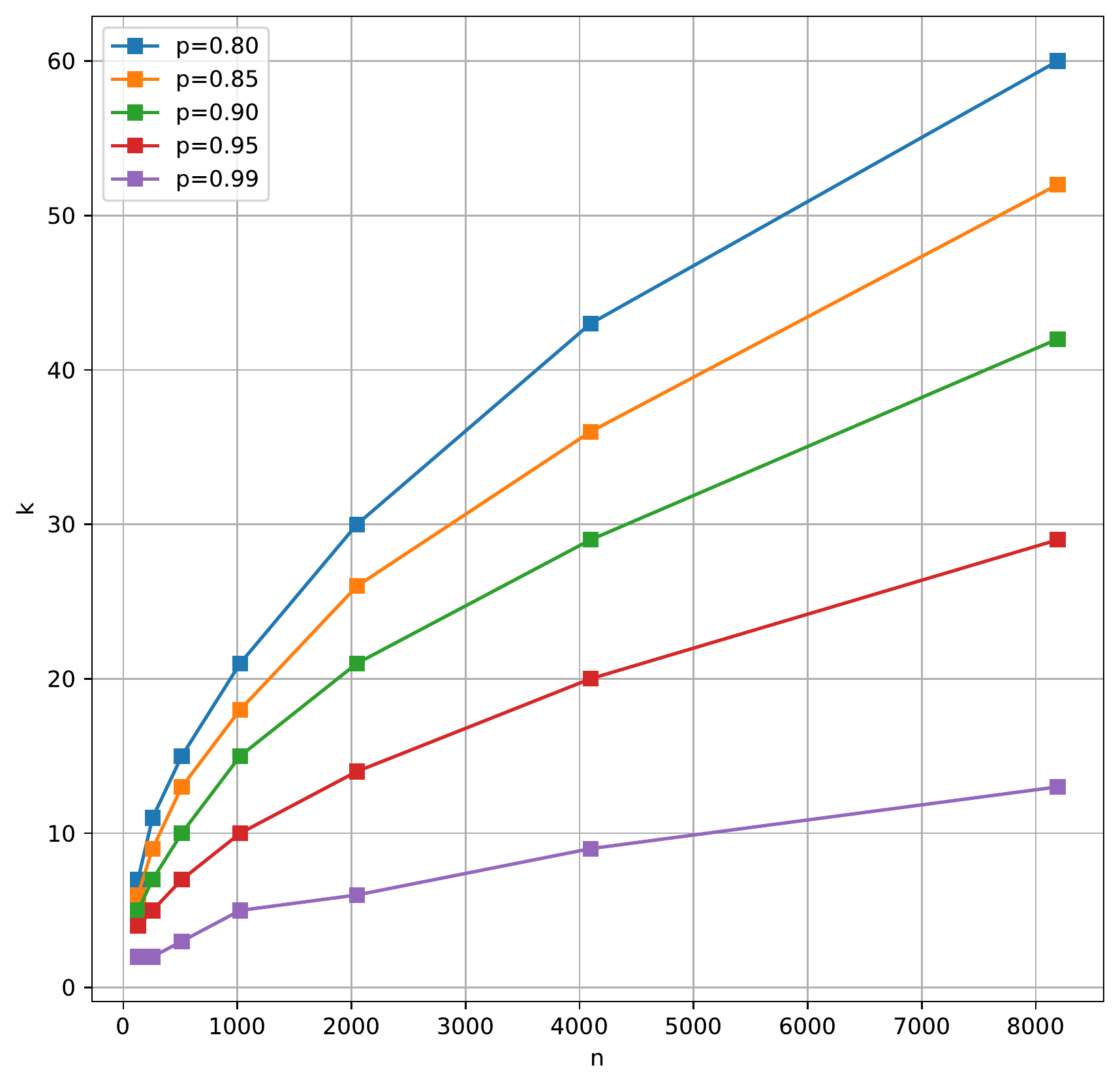}
    \caption{Maximum number of users $k$ per sub-channel such that no collisions occur with probability $p$ for varying number of resources $n$}
    \label{fig:collisions2}
\end{figure}

\begin{figure}[ptb]
    \centering
    \includegraphics[width=.49\textwidth]{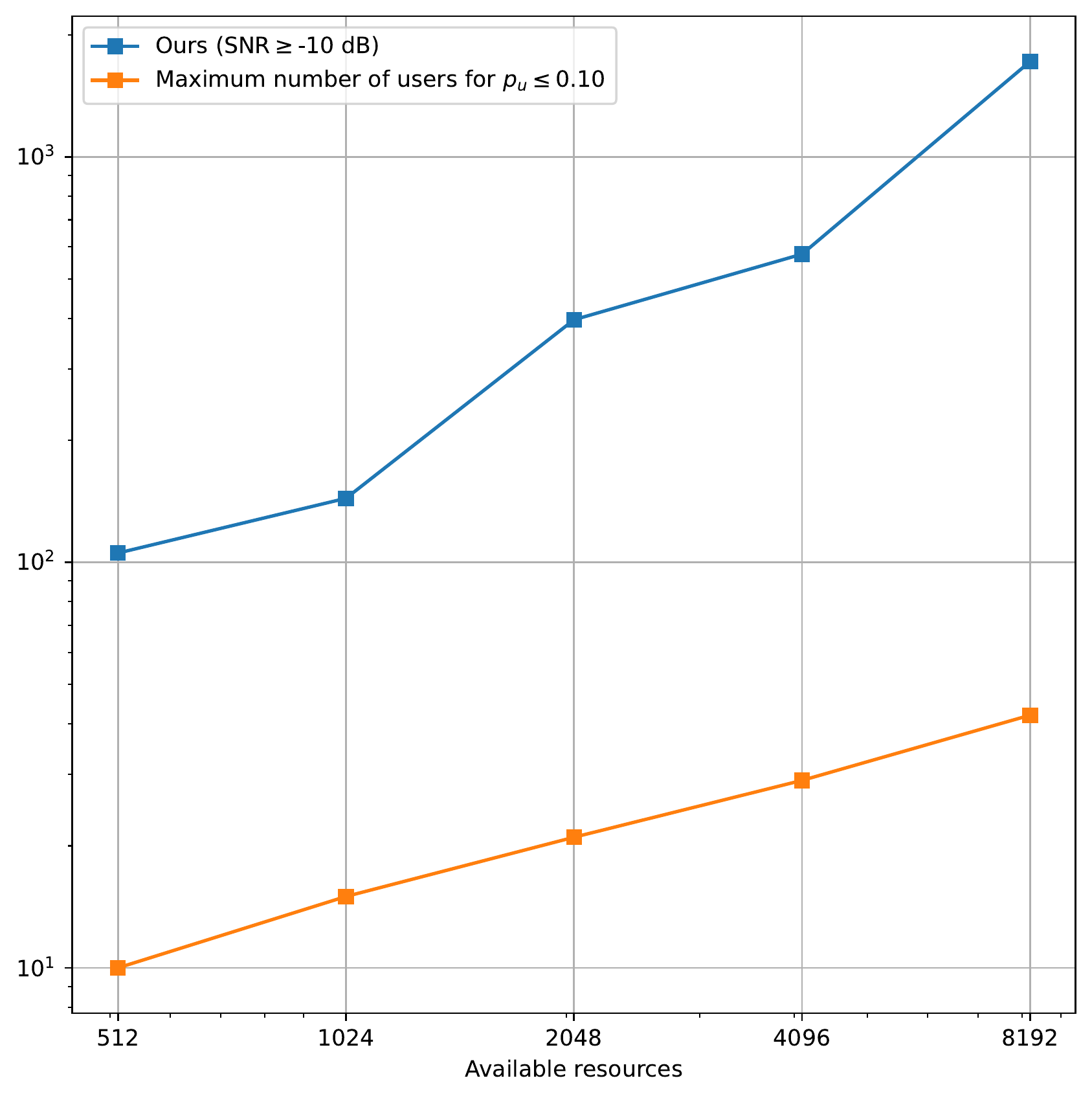}
    \caption{Average
number of supported users per system dimension over 100 Monte-Carlo trials as defined by \eqref{eq:nbrOfUsers}. Users are distributed homogeneously to the sub-channels. For comparison the number of users that can be served without sub-channeling, such that $p_u\leq 0.1$ (which also upper bounds ALOHA). The axes are scaled logarithmically.}
    \label{fig:supported_users}
\end{figure}


While the first experiment gauges the average system behaviour under ideal conditions, we conduct another simulation to investigate the ability to overload the system under more realistic premises.
In the second experiment the number of users trying to communicate over the system is not known, and the distribution of users to the sub-channels is random. In this case, the detection algorithm has no prior information on the sparsity in each sub-channel. To get a suitable estimate, the detection algorithm first thresholds each block to the assumed sparsity $k_s$ and computes each block's 2-norm. Then the block norms are clustered into 2 clusters and the blocks belonging to the smaller cluster are set to 0. 
For the simulation we set $n=2048$, the block length $s=8$ ($r=256$), and divided the image space into $c=8$ sub-channels, resulting in 2048 available resources, where 256 measurements are taken from each sub-channel. The results are averaged over 20 trials.

\begin{figure}
    \centering
    \includegraphics[width=.45\textwidth]{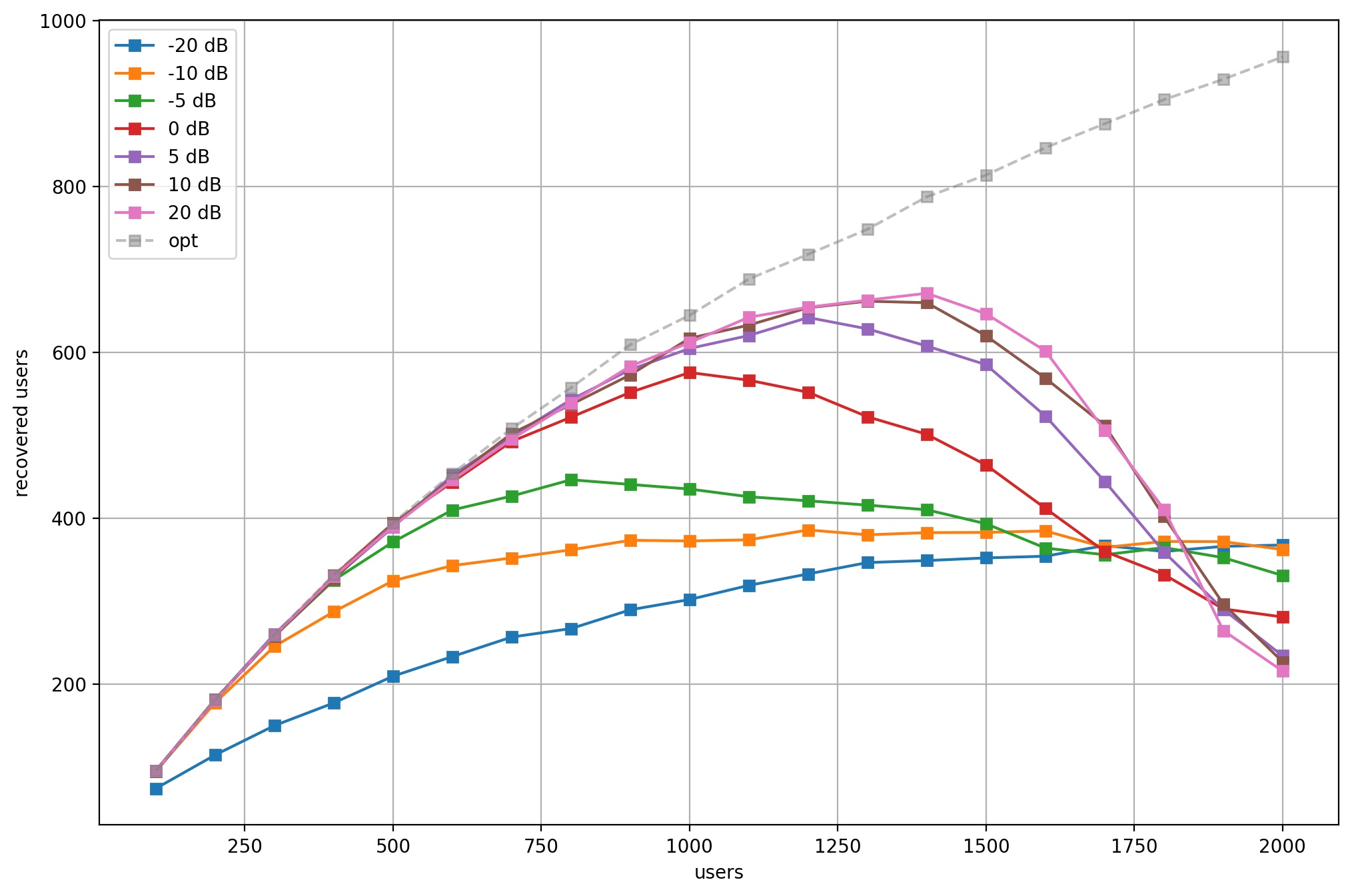}
    \caption{Total vs recovered users in the non-uniform setting. 'opt' denotes the maximum number of recoverable users, i.e. those that selected a unique pilot.}
    \label{fig:rec_users}
\end{figure}

\begin{figure}
    \centering
    \includegraphics[width=.45\textwidth]{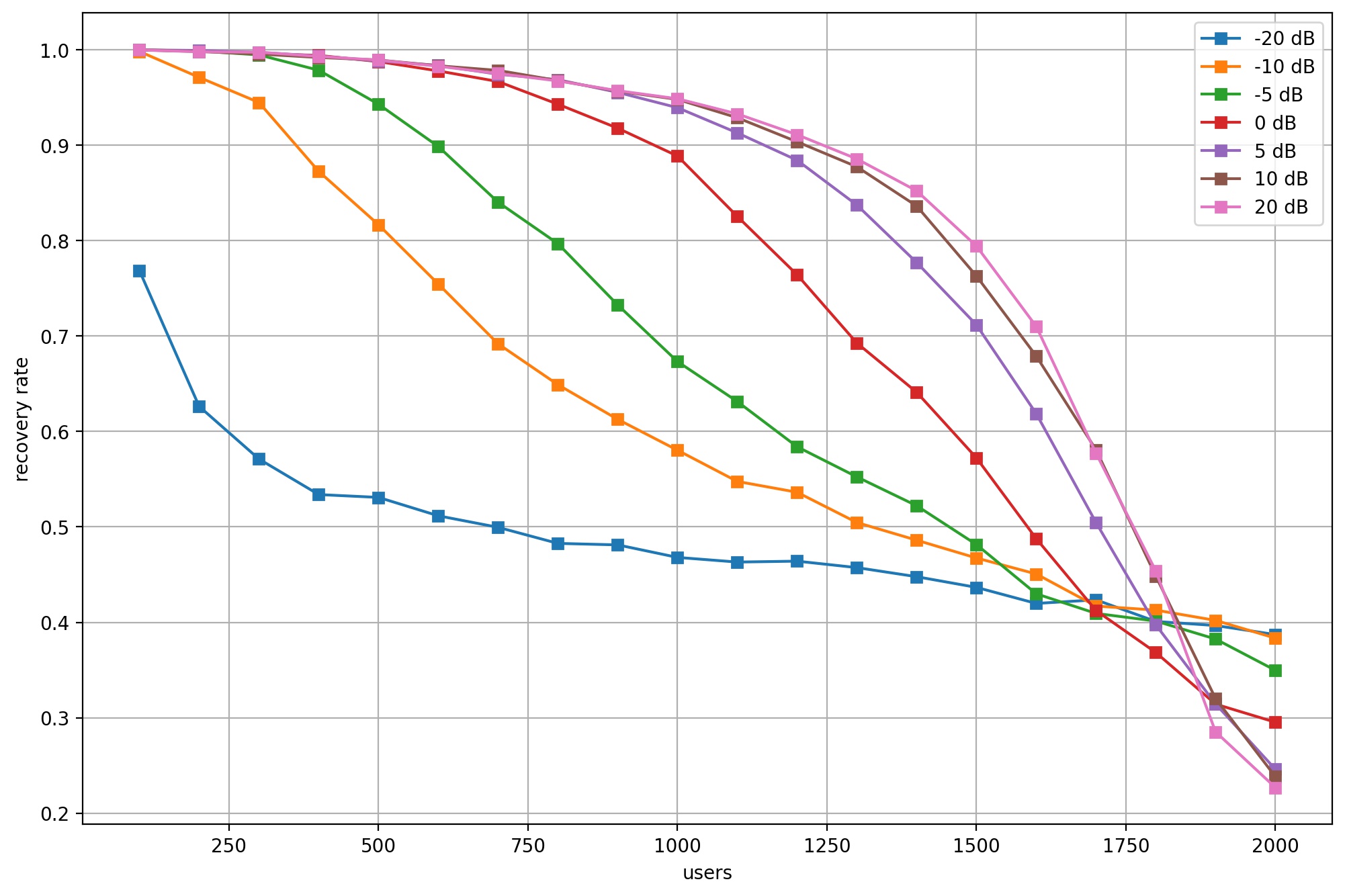}
    \caption{Recovery rates for the non-uniform setting}
    \label{fig:rec_rates}
\end{figure}

\begin{figure}
    \centering
    \includegraphics[width=.45\textwidth]{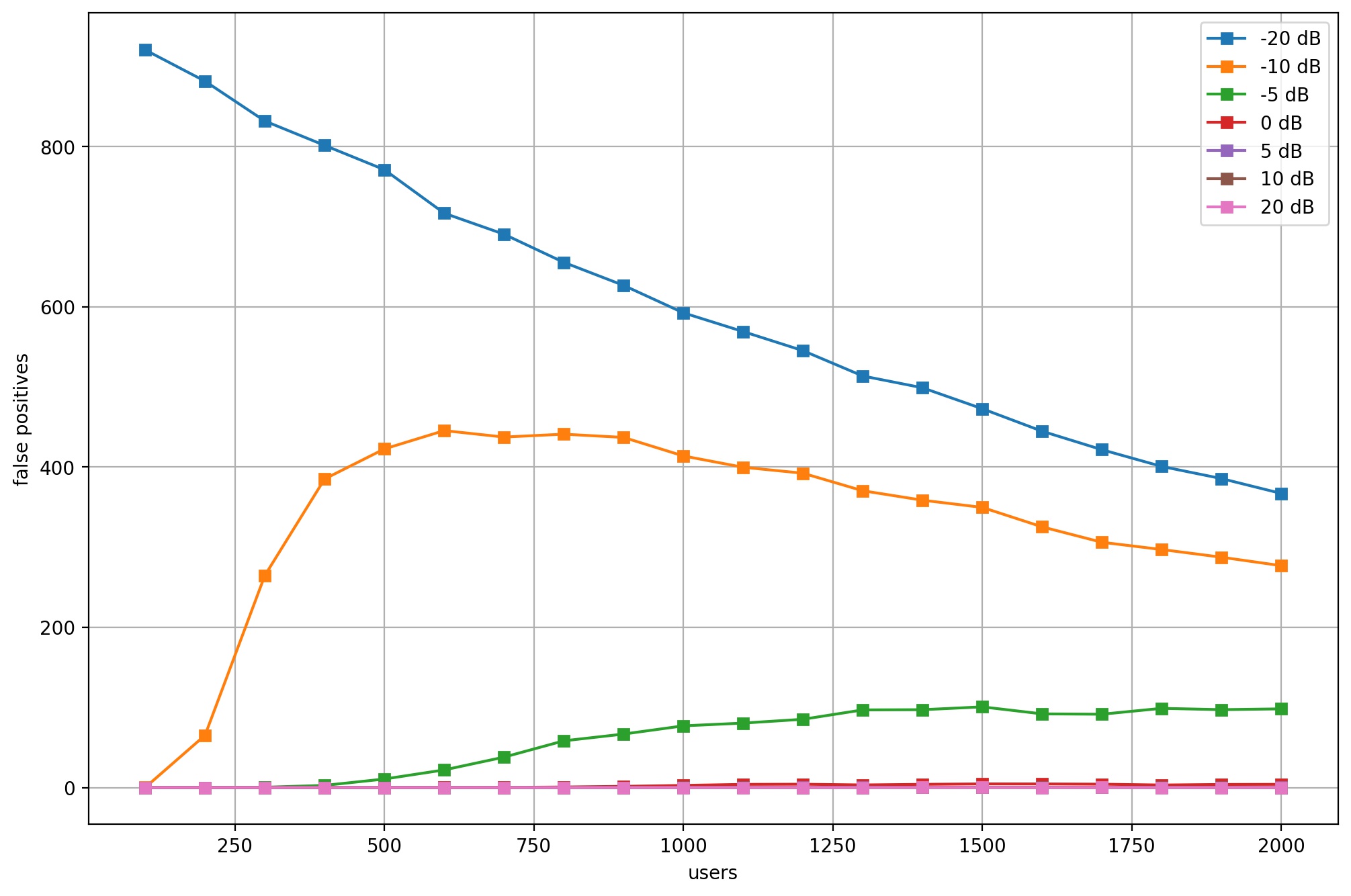}
    \caption{False positive detections vs users}
    \label{fig:fp}
\end{figure}
 {
Figure \ref{fig:rec_users} shows the number of recovered users over the number of total system users for different SNRs. In gray the average number of users that send collision-free, i.e. using a pilot sequence not chosen by any other user, is shown. This is the maximum number of users that could possibly communicate reliably. Note that until $\sim 1000$ users, and with noise level above $0$ dB (system dB!), our recovery is almost optimal. The recovery rate is depicted in Figure \ref{fig:rec_rates}. In this regime the false positive rate of our detection algorithm is also quite low as seen in Figure \ref{fig:fp}. }

 {
The recovery rate is shown to degrade gracefully with worsening SNR. When the noise power is higher than the signal power, the norms of the thresholded blocks do not differ between active blocks and pure noisy blocks, and hence the hierarchical thresholding procedure is unable to produce a reasonable estimate for the active users. In this scenario, the number of false positives can even decrease with increasing user load because many non-active blocks are classified as active and with more users altogether less non-active blocks are present. Note, however that $0$ dB in this simulation corresponds to a much lower true SNR as discussed in the beginning of this section. }
Improvements for the detection algorithm in a low SNR regime, e.g. by using HiHTP until convergence are a topic for future investigation. It is noteworthy that even in a setting where the recovery algorithm does not assume any knowledge of the number and distribution of users, the recovery is as good as in the first simulation, where exact values for the number of active resources per sub-channel are known. 
 {This experiment shows that the system can be overloaded with more users than available resources $r$ in a system without sub-channels.}

\section{Conclusion}
We designed a one-shot messaging massive random access scheme based on hierarchical compressed sensing, conducted theoretical performance analysis and demonstrated its feasibility by numerical experiments.
The proposed scheme promises huge gains in terms of number of supported users. 
Specifically, we rigorously proved that  the system can sustain effectively any overload situation, i.e. detection failure of active and non-active users can be kept below any desired threshold regardless of the number of users. The only price to pay is delay, i.e. the number of time-slots over which
cross-detection is performed. We achieved this by jointly exploring the effect of measure concentration in time and frequency and careful system parameter scaling. 
The key to proving these results were new concentration results for sequences of randomly sub-sampled DFTs detecting the sparse vectors "en bloc". Notably, since we use just mutual coherence, in principle, the results carry over to other families of matrices which are suited for the random access problem. This is a topic of further investigation.

In the numerical experiments we were able to demonstrate the overload operation. Clearly these need to be extended to test the scheme in more practical setting. Several improvements are immediate: First, one could run several iterations of the HiHTP/HiIHT algorithms to obtain better performance when the noise is strong. Second, with smaller $n$, say $n\leq 2048$, one should think of additional measures to control the sub-channel load and collisions therein.

Eventually, we have also theoretically investigated the data detection but only for the non-overload situation which gives merely stability results for the underlying underdetermined systems. It is left for future research to incorporate the effects of coding, successive interference cancellation etc. and, henceforth, carry out a complete throughput analysis for the proposed system.

\section*{Acknowledgments}
G. W. and B. G. were partially funded by German Research Foundation (DFG) grant WU 598/7-2 - SPP 1798 (Compressed Sensing in Information Processing) and project 16KISK025 in 6G Research and Innovation Cluster (6G-RIC), funded by the German Federal Ministry of Education and Research (BMBF). A.F. acknowledges support from the Wallenberg AI, Autonomous Systems and Software Program (WASP) funded by the Knut and Alice Wallenberg Foundation.

\bibliographystyle{unsrt}
\bibliography{sspa}

\section*{Author biographies}
\begin{IEEEbiographynophoto}
{Gerhard Wunder} studied electrical engineering and received his graduate degree in electrical engineering (Dipl.-Ing.) from TU Berlin with highest honors in 1999. He received the PhD degree (Dr.-Ing.) with distinction (summa cum laude) in 2003 from TU Berlin and became a research group leader at the Fraunhofer Heinrich-Hertz-Institut in Berlin. In 2007, he also received the habilitation degree (venia legendi) and became a Privatdozent (Associate Professor). In this period, he was a visiting professor at the Georgia Institute of Technology (Prof. Jayant) in Atlanta (USA, GA), and the Stanford University (Prof. Paulraj) in Palo Alto/USA (CA). In 2009 he was a consultant at Alcatel-Lucent Bell Labs (USA, NJ), both in Murray Hill (Prof. Stolyar) and Crawford Hill (Dr. Valenzuela). In 2015, he has become Heisenberg Fellow, granted for the first time to a communication engineer, and later professor heading the Heisenberg Communications and Information Theory Group at the FU Berlin. Since 2021 he is a professor for Cybersecurity and AI at FU. Recently, he has been nominated together with Dr. Müller (BOSCH Stuttgart) and Prof. Paar (Ruhr University Bochum) for the Deutscher Zukunftspreis 2017, the most prestigious innovation and research award in Germany.
\end{IEEEbiographynophoto}

\begin{IEEEbiographynophoto}
{Axel Flinth} is an assistant professor at Umeå University. He obtained his PhD Degree in 2018 from the Technische Universität Berlin. Before joining Umeå University, he was employed as as a post-doc at Université de Toulouse III - Paul Sabatier and Chalmers University, and as a guest lecturer at University of Gothenburg. 
\end{IEEEbiographynophoto}

\begin{IEEEbiographynophoto}
    {Benedikt Groß} received his diploma in mathematics from Humboldt Universität zu Berlin. He is currently a PhD student at Freie Universität Berlin under supervision of Prof. Gerhard Wunder.
\end{IEEEbiographynophoto}
\vfill
\appendix
\subsection{Two measure concentration inequalities}

 Here, we present some theoretical results we will need in our later proof. First, in the proof of Lemma \ref{lem:Psi_variation}, we will use a combination of the Hoeffding and Bernstein inequalitites. To increase readability, we formulate this combination as a separate lemma.

\begin{lemma} \label{lem:mixedchaos}
    Let $(X_q)_{q \in [d]}$ be independent real random variables of the form
    \begin{align*}
        X_q = c_q + g_q + x_q
    \end{align*}
    where $c_q$ are constants, $g_q$ are subgaussians with subgaussian norms $\gamma_q$, and $x_q$ are subexponentials with subexponential norms $\xi_q$. Consider the random variable
        \begin{align*}
            V = \sum_{p \in [d]}  X_p
        \end{align*}
        Then 
        \begin{align*}
            &\mathbb{P}\left( \abs{ V - \erw{V} }>\theta + \rho\right)  \leq  2\exp \left(-\tfrac{\kappa\theta^2}{\sum_p \gamma_p^2}\right)  \\
            & \quad + 2\exp \left(- \min\left(\tfrac{\kappa\rho^2}{\sum_p 2\xi_p^2} , \tfrac{\kappa\rho}{\sup_p  \xi_p}\right)\right)
        \end{align*}
        where $\kappa$ is a numerical constant.
\end{lemma}
\begin{proof} 
        As mentioned, this is a simple combination of the Hoeffding inequality for subgaussians and the Bernstein inequality. The first namely states
        \begin{align*}
            \mathbb{P}\left(\abs{\sum_{p \in [d]} g_p - \mathbb{E}\big(\sum_{p \in [d]} g_p\big)}> \theta \right) \leq 2\exp \left(-\tfrac{\kappa\theta^2}{\sum_p \gamma_p^2}\right),
        \end{align*}
        and the second 
        \begin{align*}
            &\mathbb{P}\left(\abs{\sum_{p \in [d]} x_p - \mathbb{E}\big(\sum_{p \in [d]} x_p\big)}> \rho \right) \leq \\
            & \quad 2\exp \left(- \kappa \min\left(\tfrac{\rho^2}{\sum_p \xi_p^2} , \tfrac{\rho}{\sup_p  \xi_p}\right)\right).
        \end{align*}
    \end{proof}

In the proof of Lemma \ref{lem:psivariation}, the above will not suffice, because the expression we are going to try to control is a fourth-order polynomial in the noise vector $z$ (and hence cannot be conceived as a sum of Gaussian and subexponential variables only). For this, we will need the following more powerful result.

\begin{theorem} \emph{(Simplified version of \cite[Theorem 1.4]{Adamczak2013ConcentrationIF}.)} \label{th:polyconc}
    Let $\zz=(z_i)_{i \in [n]}$ be a random vector with independent, complex i.i.d Gaussian entries, with variance $\varsigma^2$. For a $k$-multilinear form $M$, let $\norm{M}_F$ denote the norm
    \begin{align*}
        \norm{M}_F^2 = \sum_{i_0, \dots, i_{k-1}} \abs{M(\ee_{i_0}, \dots, \ee_{i_{k-1}})}^2,
    \end{align*}
    where $e_{i}$ is the $i$:th unit vector. Then, for every polynomial $f: \C^n \to \C$ of degree $d$, we have
    \begin{align*}
        &\mathbb{P}\left(\abs{f(\zz)-\erw{f(\zz)}} >\theta \right) \\
        &\quad \leq \exp( - \kappa \min_{1\leq k \leq d} \min_{1\leq j \leq k} \left(\tfrac{\theta}{\lambda_k})\right)^{\sfrac{2}{k}}),
    \end{align*}
    where $\kappa$ is a numerical constant and $\lambda_k$ is defined as
    \begin{align*}
        \lambda_k = \varsigma^k \norm{\erw{f^{(k)}(\zz)}}_F,
    \end{align*}
    where $f^{(k)}$ denotes the $k$:th derivative of $f$.
\end{theorem}

A few remarks are in order. First, the theorem in \cite{Adamczak2013ConcentrationIF} is formulated and proved for real variables -- however, since a Gaussian in $\C^n$ can be interpreted as a real Gaussian in $\R^{2n}$, and we always can split the real and imaginary value of $f(\zz)$, the theorem goes through also for complex variables (possibly with slightly worse implicit constant). Secondly, the theorem actually holds for any subgaussian distributions of the $z_i$ -- however, we only need the Gaussian case. Finally, the bound claimed in the cited source looks much more complicated, since it involves other norms of the multilinear forms. However, as is remarked earlier in the paper, those norms are all bounded by $\norm{\,\cdot \,}_F$, so that the above theorem still is true.

\subsection{A few deterministic  bounds.}
Next, let us bound a few expressions involving $\AA_\omega^*\AA_p$ under a coherence assumptions. In contrast to many of the other bounds we will prove, these are completely deterministic. We will need them all in the coming sections.

A notational remark is in order: For a vector $\xx \in \C^p$, we will in the following refer to its $i$:th element either by $x_i$ or $\xx(i)$, depending on what is more convenient. For instance, in instances where the vector itself is endowed with a sub-index, we will opt for the second alternative.

\begin{lemma} \label{lem:CP}
    For $\AA$, $\omega$ and $p$ arbitrary, define $\CC_p = \AA_\omega^*\AA_p$. Under the assumption that the coherence of $\AA$ is smaller than $\tfrac{\tau}{\sqrt{k_uk_s^2}}$, 
    \begin{align*}
        \abs{\sum_{p \in [r]} \norm{\CC_p\hh_p}^2 - \norm{\hh_\omega}^2} &\leq \tfrac{2\tau}{\sqrt{k_u}}\norm{\hh_\omega}\norm{\hh} + \tfrac{\tau^2}{k_u}\norm{\hh}^2 \\
        \sum_{p \in [r]} \norm{\CC_p^*\CC_p \hh_p}^2  &\leq (k_s + \tfrac{\tau^2s}{k_uk_s})\max(1,\tau)^2\mathfrak{h}^2 \\
        \sum_{p \in [n_s]} \norm{\CC_p}_F^2 &\leq (k_s + \tfrac{\tau^2n}{k_uk_s}) \\
         \sup_{p \in [n_s]} \norm{\CC_p}_F^2 &\leq (k_s + \tfrac{\tau^2s}{k_uk_s}) \\
    \end{align*}
    where the notation $\mathfrak{h}$ was defined in Lemma \ref{lem:Psi_variation}. 
\end{lemma}
\begin{proof}
We begin by fixing $p$ and estimating $\norm{\CC_p\hh_p}^2$. Let $I_p$ be the indices in $[n]$ corresponding to the $p$:th block.
    \begin{align*}
        \norm{\CC_p\hh_p}^2 &= \sum_{\ell \in \omega} \abs{\sprod{\aa_\ell,\AA_p\hh_p}}^2 \\
        &= \sum_{\ell \in \omega} \sum_{k,j \in I_p} \sprod{\aa_k,\aa_\ell}\sprod{\aa_\ell,\aa_j} \overline{\hh_p(k)}\hh_p(j).
    \end{align*}
    We distinguish two cases
    
    \underline{$\omega \cap I_p = \emptyset$} In this case, we have $\abs{\sprod{\aa_k,\aa_\ell}\sprod{\aa_\ell,\aa_j}}\leq\tfrac{\tau^2}{k_uk_s^2}$ for all values of $k,j$ and $\ell$. Consequently
    \begin{align*}
        &\sum_{\ell \in \omega} \sum_{k,j \in I_p} \sprod{\aa_k,\aa_\ell}\sprod{\aa_\ell,\aa_j} \overline{\hh_p(k)}\hh_p(j) \\
        &\quad \leq \tfrac{\tau^2}{k_uk_s} \sum_{k,j \in I_p} \overline{\hh_p(k)}\hh_p(j)  \leq \tfrac{\tau^2}{k_u} \norm{\hh_p}^2,
    \end{align*}
    where we in the final step utilized the Cauchy-Schwarz inequality together with the fact that $\hh$ is $(k_u,k_s)$-sparse (so that $\hh_p$ is $k_s$-sparse).
    
    \underline{$\omega\sse I_p$} Let us divide the inner sum into the index pairs where $k=j$ and the ones where $k\neq j$.
    \begin{align*}
        \sum_{\ell \in \omega} \sum_{k \in I_p} \abs{\sprod{\aa_k,\aa_\ell}}^2\abs{\hh_p(k)}^2 
    \end{align*}
    We have
    \begin{align*}
       &\sum_{k \in I_p} \abs{\sprod{\aa_k,\aa_\ell}}^2\abs{\hh_p(k)}^2 \\
       &\quad = \abs{\hh_p(\ell)}^2 + \sum_{k \neq \ell }\abs{\sprod{\aa_k,\aa_\ell}}^2 \abs{\hh_p(\ell)}^2
    \end{align*}
    Summing this equality over $\ell$ yields
    \begin{align*}
         &\sum_{\ell \in \omega} \sum_{k \in I_p} \abs{\sprod{\aa_k,\aa_\ell}}^2\abs{\hh_p(k)}^2 = \\
         &\qquad  \norm{(\hh_p)_\omega}^2 + \sum_{\ell\in \omega}\sum_{k \neq \ell }\abs{\sprod{\aa_k,\aa_\ell}}^2 \abs{\hh_p(k)}^2.
    \end{align*}
    The rest term can be bounded by
    \begin{align*}
        \tfrac{\tau^2}{k_uk_s} \norm{\hh_p}^2,
    \end{align*}
    since $\abs{\sprod{\aa_k,\aa_\ell}}^2  \leq \tfrac{\tau^2}{k_uk_s^2}$ for $k\neq \ell$, and $\abs{\omega}=k_s$.
    
    We continue with the terms for which $k \neq j$. In this case, we can estimate
    \begin{align*}
        &\sum_{\ell \in \omega} \abs{\sprod{\aa_k,\aa_\ell}\sprod{a_\ell,a_j}} \\
        &\quad \leq \tfrac{\tau}{\sqrt{k_uk_s^2}}(\mathds{1}_\omega(k) + \mathds{1}_\omega(j)) +  \tfrac{\tau^2}{k_uk_s}.
    \end{align*}
    We furthermore have, due to the Cauchy-Schwartz inequality together with the $k_s$-sparsity of the $\hh_p$,
    \begin{align}
        \sum_{k,j \in I_p} \mathds{1}_\omega(k) \abs{\hh_p(k)\hh_p(j)} &\leq k_s \norm{(\hh_p)_\omega}\norm{\hh_p} \\
        \sum_{k,j \in I_p}  \abs{\hh_p(k)\hh_p(j)} &\leq k_s \norm{\hh_p}^2
        \label{eq:akbl}
    \end{align}
    Consequently, the sum as a whole can be estimated with
    \begin{align*}
        \tfrac{2\tau}{\sqrt{k_u}}\norm{(\hh_p)_\omega}\norm{\hh_p}  +  \tfrac{\tau^2}{k_u} \norm{\hh_p}^2.
    \end{align*}
    Hence,
    \begin{align*}
        &\abs{\norm{\CC_p\hh_p}^2- \norm{(\hh_p)_\omega}^2} \\
        &\quad \leq \tfrac{2\tau}{\sqrt{k_u}}\norm{(\hh_p)_\omega}\norm{\hh_p}  +  \tfrac{\tau^2}{k_u} \norm{\hh_p}^2.
    \end{align*}
    Summing this inequality over $p$ yields the first inequality.
    
    We move on to $\norm{\CC_p}_F^2$.
    \begin{align*}
        \norm{\CC_p}_F^2 = \sum_{\ell \in \omega, k \in I_p} \abs{\sprod{\aa_\ell,\aa_k}}^2.
    \end{align*}
    We again distinguish between two cases.
    
    \underline{ $\omega \sse I_p$} In this case, for each value of $\ell$, $\abs{\sprod{\aa_\ell,\aa_k}}^2=1$ for exactly one value of $k$ ($k=\ell$), but is else smaller than $\tfrac{\tau^2}{k_uk_s^2}$. Therefore
    \begin{align*}
        \sum_{\ell \in \omega, k \in I_p} \abs{\sprod{\aa_\ell,\aa_k}}^2 \leq \sum_{\ell\in \omega} 1 + \tfrac{\tau^2}{k_uk_s^2} \cdot \abs{I_p} \leq k_s + \tfrac{\tau^2 \abs{I_p}}{k_uk_s}.
    \end{align*}
    
    \underline{ $\omega \cap I_p= \emptyset$} In this case, $\abs{\sprod{\aa_\ell,\aa_k}}^2$ is always smaller than $\tfrac{\tau^2}{k_uk_s^2}$. Therefore
    \begin{align*}
         \sum_{\ell \in \omega, k \in I_p} \abs{\sprod{\aa_\ell,\aa_k}}^2 \leq  \tfrac{\tau^2\abs{I_p}}{k_uk_s}.
    \end{align*}
    This implies the last estimate. Also, since $\omega \sse I_p$ for exactly one value of $p$, we get
    \begin{align*}
        \sum_{p} \norm{\CC_p}_F^2 \leq k_s + \tfrac{\tau^2n}{k_uk_s},
    \end{align*}
    i.e., the third.
    
    We move on to the final inequality. This however easily follows from the previous ones:
    \begin{align*}
        &\sum_{p \in [r]} \norm{\CC_p^*\CC_p\hh_p}^2 \leq \sum_{p \in [r]}\norm{\CC_p}_F^2\norm{\CC_p\hh_p}_2^2 \\
        &\quad \leq \sum_{p \in [r]}\left(k_s + \tfrac{\tau^2s}{k_uk_s}\right)\norm{\CC_p\hh_p}_2^2 \\
        &\leq  \left(k_s + \tfrac{\tau^2s}{k_uk_s}\right) \left(\norm{\hh_\omega}^2 + 2 \tfrac{2\tau}{\sqrt{k_u}}\norm{\hh_\omega}\norm{\hh} + \tfrac{\tau^2}{k_u}\norm{\hh}^2\right) \\
        &\leq         
        \left(k_s + \tfrac{\tau^2s}{k_uk_s}\right)\max(1,\tau)^2\mathfrak{h}^2
    \end{align*}
\end{proof}

\begin{lemma} \label{lem:omegasum}
    Let $p \neq q$ and let $k\in I_p$ and $\ell \in I_q$ and $\omega$ $(1,k_s)$-sparse be arbitrary. Under the assumption that the coherence of the matrix $A$ is bounded by {$\tfrac{\tau^2}{\sqrt{k_uk_s^2}}$}, we have 
    \begin{align*}
        &\sum_{j \in \omega} \abs{\sprod{\aa_\ell,\aa_j}\sprod{\aa_j,\aa_k}} \\
        &\qquad \leq \left(\mathds{1}_\omega(k) + \mathds{1}_\omega(\ell)\right)\tfrac{\tau}{\sqrt{k_uk_s^2}} + \tfrac{\tau^2}{k_uk_s}.
    \end{align*}
\end{lemma}
\begin{proof}
    Since $p\neq q$, there are three cases: $\ell \in \omega$, $k \in \omega$ and $\ell,k \notin \omega$. Let us treat them separately
    
    \underline{$\ell \in \omega$} Note that since $p\neq q$, we then necessarily have $\omega \cap I_p = \emptyset$. Hence, $\abs{\sprod{\aa_j,\aa_k}}\leq \tfrac{\tau}{\sqrt{k_uk_s^2}}$ for all $j$. Furthermore, $\sprod{\aa_\ell,\aa_j}$ is equal to $1$ exactly when $j=\ell$, and else also smaller than $\tfrac{\tau}{\sqrt{k_uk_s^2}}$ in absolute value. We can therefore estimate
    \begin{align*}
        \sum_{j \in \omega} \abs{\sprod{\aa_\ell,\aa_j}\sprod{\aa_j,\aa_k}} \leq \tfrac{\tau}{\sqrt{k_uk_s^2}} + \tfrac{\tau^2}{k_uk_s}
    \end{align*}
        \underline{$k \in \omega$} In this case, we can instead immediately estimate $\abs{\sprod{\aa_\ell,\aa_j}}\leq \tfrac{\tau}{\sqrt{k_uk_s^2}}$ for all $j$, and conclude that $\sprod{\aa_j,\aa_k}$ is equal to $1$ exactly when $j=k$, and else also smaller than $\tfrac{\tau}{\sqrt{k_uk_s^2}}$ in absolute value. All in all, the sum can be bounded by
        \begin{align*}
         \sum_{j \in \omega} \abs{\sprod{\aa_\ell,\aa_j}\sprod{\aa_j,\aa_k}} \leq \tfrac{\tau}{\sqrt{k_uk_s^2}} + \tfrac{\tau^2}{k_uk_s}
        \end{align*}

    \underline{$\ell,k \notin \omega$} In this case, all scalar products are smaller than $\tfrac{\tau}{\sqrt{k_uk_s^2}}$, and we may simply estimate
        \begin{align*}
         \sum_{j \in \omega} \abs{\sprod{\aa_\ell,\aa_j}\sprod{\aa_j,\aa_k}} \leq \tfrac{\tau^2}{k_uk_s}.
        \end{align*}
        
    The bound follows.
\end{proof}

\section{The deviation lemmata} \label{sec:devlem}

We now arrive at the lemmata we claimed in the main text. Let us start with the one analysing the deviation of $\Psi(\AA,\BB,\hh+\zz)$ from $\Psi(\AA,\BB,\hh)$. Note that since $\zz$ is independent of all other random variables, we may treat the latter as constant in our considerations.

\begin{proof}[Proof of Lemma \ref{lem:Psi_variation}]

    Note that $\Psi(\AA,\BB,\hh)$ is a mean of independent variables
    \begin{align*}
        U^i = \sum_{p} \sprod{\AA_p^i(\hh_p^i+\zz_p^i),\BB_p^i(\hh_p^i+\zz_p^i)}.
    \end{align*}
    We will first 
    prove a high-probability bound on all of them. We may then apply Hoeffding to get the final concentration. To further ease the notation, let us also drop the index $i$ in passing.

    Let us first note that for $\uu,\ww$ arbitrary
    \begin{align*}
        \sprod{\AA_p\uu_p,\BB_p\ww_p} &= \sprod{\AA_\omega^*\AA_p \uu_p, \AA_\omega^*\AA_p \ww_p} \\
        &= \sprod{\CC_p\uu_p, \CC_p \ww_p},
    \end{align*}
    where we used the previously defined notation $\CC_p = \AA_\omega^*\BB_p$. Hence,
    \begin{align*}
        U =& \sum_{p}  \sprod{\CC_p(\hh_p+\zz_p),\CC_p(\hh_p+\zz_p)} - \norm{\hh_\omega}^2 \\
        =& \sum_{p} \sprod{\CC_p\hh_p,\CC_p\hh_p} \\
        &+ \sum_p 2\mathrm{Re} (\sprod{\CC_p\hh_p,\CC\zz_p}) + \sprod{\CC_p\zz_p,\CC_p\zz_p}) 
    \end{align*}
    Here, $X_p  = \sprod{\CC_p\hh_p,\CC_p\hh_p} + 2\mathrm{Re} (\sprod{\CC_p\hh_p,\CC_pzz_p}) + \sprod{\CC_p\zz_p,CC_p\zz_p}) $ are variables 
    with 
    \begin{align*}
        \gamma_p^2 &= 4\erw{\abs{\sprod{\CC_p^* \CC_p\hh_p,\zz_p}}^2} = 4 \tfrac{\sigma^2m}{n^2}\norm{\CC_p^*\CC_p\hh_p}^2 \\
        \xi_p &= \norm{\CC_p\zz_p}_{\mathrm{subg.}}^2 = \tfrac{\sigma^2m}{n^2} \norm{\CC_p}_F^2.
    \end{align*} 
   By Lemma \ref{lem:CP}, we have 
    \begin{align*}
        \sum_{p}  \gamma_p^2 &= \tfrac{4\sigma^2 m}{n^2} \sum_{p} \norm{\CC_p^*\CC_p\hh_p}^2 \\
        &\leq \tfrac{4\sigma^2 m}{n^2}(k_s + \tfrac{\tau^2s}{k_uk_s}) \max(1,\tau)^2\mathfrak{h}^2 \\
        \sum_{p} \xi_p^2 &= \sum_p \tfrac{\sigma^4m^2}{n^4} \norm{\CC_p}_F^4 \\
        &\leq  \tfrac{\sigma^4m^2}{n^4} ( k_s + \tfrac{\tau^2s}{k_uk_s})\sum_p  \norm{\CC_p}_F^2   \\
        &\leq \tfrac{\sigma^4m^2}{n^4} ( k_s + \tfrac{\tau^2s}{k_uk_s})(k_s + \tfrac{\tau^2n}{k_uk_s}) \\
        &\leq \tfrac{\sigma^4m^2}{n^2} ( 1 + \tfrac{\tau^2}{k_uk_s})^2\\
        \sup_{p}  \xi_p \leq & \sup \tfrac{\sigma^2m}{n^2} \norm{\CC_p}_F^2 \leq \tfrac{\sigma^2m}{n^2}(k_s + \tfrac{\tau^2s}{k_uk_s}) \\
        &\leq \tfrac{\sigma^2m}{n}(1 + \tfrac{\tau^2}{k_uk_s}) 
    \end{align*}
    Therefore, with the choices
    \begin{align*}
        \theta &= \left(\log(t)\tfrac{4\sigma^2 m}{n}(k_s + \tfrac{\tau^2s}{k_uk_s})\max(1,\tau)^2\right)^{\sfrac{1}{2}}\mathfrak{h} \\
        \rho &=\log(t) \tfrac{\sigma^2 m}{n}(1 + \tfrac{\tau^2}{k_uk_s}\max(1,\tau)^2),
    \end{align*}
    we obtain that
    \begin{align*}
        &\abs{V_0 - \mathbb{E}(V_0)} \\
        & \quad \leq   \left(\log(t)\tfrac{4\sigma^2 m}{n}(1 + \tfrac{\tau^2}{k_uk_s})\max(1,\tau)^2\right)^{\sfrac{1}{2}}\mathfrak{h} \\
        &\qquad + \log(t) \tfrac{\sigma^2 m}{n}(1+\tau^2)\max(1,\tau^2).
    \end{align*}
    for all $i$ with a failure probability smaller than $4te^{-\kappa\log(t)n}= 4t^{1-\kappa n}$. We (crudely) estimated $k_u^{-1}k_s^{-1}\leq 1$. 
    
    Now, applying Hoeffding on the event that the above bound holds yields that
    \begin{align*}
        \abs{ \tfrac{1}{t}\sum_{i \in [t]} (V_0^i- \erw{V^i_0})} \leq \mathfrak{h} \tfrac{2\sigma m^{\sfrac{1}{2}}}{n^{\sfrac{1}{2}}} +  \tfrac{\sigma^2 m}{n}
    \end{align*}
    with a failure probability smaller than $2 \exp( - \tfrac{\kappa t}{\log(t)^2(1 + \tfrac{\tau^2}{k_uk_s})\max(1,\tau)^2})$.

    Now it is only left to note that
    \begin{align*}
        \tfrac{1}{t}\sum_{i \in [t]}  \erw{V^i} &= \tfrac{1}{t} \sum_{i \in [t]} \sum_{p} \sprod{\AA_p^i\hh_p,\BB_p^i\hh_p} \\
        &= \Psi(\AA,\BB,\hh)
   \end{align*}
   The proof has been finished.
    
\end{proof}

We move on to the lemma controlling the deviation of $\psi(\AA,\BB,\hh+\zz)$ from $\psi(\AA,\BB,\hh)$. As stated earlier, we can treat all terms but $\zz$ as constant in these considerations. As such, the expression we are trying to control is a fourth order polynomial in the Gaussian vector $\zz$. Hence, we will be able to control it with Theorem \ref{th:polyconc}. To do so, let us first prove some auxiliary results about the expectation of derivatives of certain polynomials in Gaussians. 

\begin{lemma} \label{lem:derivs}
    Let $\zz$ be a $d$-dimensional random vector with independent, centered Gaussian entries, with variance $\varsigma^2$. Let further $\alpha$ be a scalar, $\beta$ be a (real) linear form of the form $\beta(\xx) = \sprod{\bb,\xx} + \langle{\xx,\tilde{\bb}}\rangle$ and $\GGamma$ be a matrix with $\Gamma_{ii}=0$ for all $i$. Consider the polynomial
    \begin{align*}
        \pi(\zz) = \abs{\alpha +\beta(\zz) + \sprod{\zz,\GGamma \zz}}^2.
    \end{align*}
    Then
    \begin{align*}
        \erw{\pi(\zz)}= & \abs{\alpha}^2 + \varsigma^2(\norm{\bb}^2+\norm{\tilde{\bb}}^2) + \varsigma^4 \norm{\GGamma}_F^2 \\
        \erw{\pi'(\zz)(\xx)} =& 2  \mathrm{Re} \left(\varsigma^2\overline{\beta(\GGamma \xx_0)} + \overline{\alpha}\beta(\xx_0)\right) \\ 
        \erw{\pi''(\zz)(\xx)} =& 2\varsigma^2 \mathrm{Re}(\sprod{\GGamma \xx_0,\GGamma \xx_1} + \sprod{\GGamma^* \xx_1,\GGamma^* \xx_0}) \\
        &+ \overline{\beta(\xx_0)}\beta(\xx_1)  + 4\mathrm{Re}(\overline{\alpha} \gamma(\sprod{\xx_1, \xx_0})) \\
        \erw{\pi^{(3)}(\zz)(\xx)} =& 4\sum_{i\in [3]} \mathrm{Re}(\overline{\sprod{\beta,x_{i+2}}}\gamma({\xx_{i+1}, \xx_i}) \\
        \erw{\pi^{(4)}(\zz)(\xx)} =& \sum_{\pi \in S_4} \mathrm{Re}(\gamma({\xx_{\pi(0)}, \xx_{\pi(1)}})\gamma({\xx_{\pi(2)}, \xx_{\pi(3)}})),
    \end{align*}
    where we defined the shorthand $\gamma(\xx,\yy) = \sprod{\xx,\GGamma \yy} + \sprod{\yy,\GGamma \xx}$
\end{lemma}

\begin{proof}

    Let's use the additional  short hands $\delta(Z) = \alpha + \beta(\zz) + \sprod\zz{,\Gamma \zz}$, 
    and $\varepsilon(\zz,X\xx_0) = \beta(\xx_0) +\gamma(\zz,\xx_0)$. We have
    \begin{align*}
        \pi'(\zz)\xx_0 =& 2\mathrm{Re}(\overline{\delta(\zz)}\varepsilon(\zz,\xx_0))\\
        \pi''(\zz)(\xx_0)=& 2\mathrm{Re}(\overline{\varepsilon(\zz,\xx_0)}\varepsilon(\zz,\xx_1)) \\
        &+ 4\mathrm{Re}(\overline{\delta(\zz)}\gamma(\xx_0,\xx_1)) \\
        \pi^{(3)}(\zz)(\xx) =& 2\sum_{i\in [3]} \mathrm{Re}(\overline{\varepsilon(\zz,X_{i+2})}\gamma(\xx_{i+1},\xx_i)) \\
        \pi^{(4)}(Z)(\xx) =& \tfrac{2}{4}\sum_{\pi \in S_4} \mathrm{Re}(\overline{\gamma({\xx_{\pi(0)}, \xx_{\pi(1)}})}\gamma({\xx_{\pi(2)}, \xx_{\pi(3)}}))
    \end{align*}
    The statement about the fourth derivative follows immediately, since it is constant. As for the other terms, we now only need to calculate the expected value. Using the fact that $\erw{\varepsilon(\zz,\xx_0)}= \beta(\xx_0)$, we get
    \begin{align*}
        \erw{\pi^{(3)}(\zz)(\xx)} =& 4\sum_{i\in [3]} \mathrm{Re}(\overline{\beta(\xx_{i+2})}\gamma({\xx_{i+1},\Gamma \xx_i}) )
    \end{align*}
    We also have
    \begin{align*}
        \erw{\delta(\zz)} &= \alpha + \erw{\beta(\zz)} + \erw{\sprod{\zz,\GGamma \zz}} \\
        &= \alpha + \varsigma^2\mathrm{tr}(\GGamma) = \alpha.
    \end{align*}
    We here used the symmetry of $\zz$, and the assumption of $\Gamma$ having a zero diagonal. Consequently
    \begin{align*}
        \erw{\overline{\delta(\zz)}\gamma(\xx_1,\xx_0)} = \overline{\alpha} \gamma(\xx_1,\xx_0)
    \end{align*}
    We move on to the first term in the second derivative. We have
    \begin{align*}
        \overline{\varepsilon(\zz,\xx_0)}\varepsilon(\zz,\xx_1) &= \overline{\beta(\xx_0)}\beta(\xx_1) + \overline{\gamma(\zz,\xx_0)}\beta(\xx_1)\\
        &+ \overline{\beta(\xx_0)}\gamma(\zz,\xx_1) + \overline{\gamma(\zz,\xx_0)}\gamma(\zz,\xx_1).
    \end{align*}
    The two terms with $\zz$ appearing linearly vanish in expectation (since $\zz$ is centered). The other nonconstant term equals
    \begin{align} \label{eq:term}
        &\sprod{\GGamma \xx_0, \zz}\sprod{\zz,\GGamma \xx_1} + \sprod{\GGamma \zz, \xx_0}\sprod{\zz,\GGamma \xx_1} \\
        & \quad + \sprod{\GGamma\xx_0, \zz}\sprod{\xx_1,\GGamma \zz} \sprod{\GGamma \zz, \xx_0}\sprod{\xx_1,\GGamma \zz} \nonumber
    \end{align}
   The following two equalities hold
    \begin{align*}
        \erw{\sprod{\uu,\zz}\sprod{\zz,\ww}} &= \uu^*\erw{\zz\zz^*}\ww = \varsigma^2 \sprod{\uu,\ww} \\
        \erw{\sprod{\zz,\uu}\sprod{\zz,\ww}} &= 0.
    \end{align*}
    The latter equality can be seen through direct calculation, or by the fact that $i\zz$ is identically distributed to $\zz$. Hence
    \begin{align*}
         \erw{\sprod{\zz,\uu}\sprod{\zz,\ww}} &=  \erw{\sprod{i\zz,\uu}\sprod{i\zz,\ww}} \\
         &= - \erw{\sprod{\zz,\uu}\sprod{\zz,\ww}}.
    \end{align*}
    Consequently
    \begin{align*}
        \erw{\eqref{eq:term}} = \varsigma^2(\sprod{\GGamma \xx_0,\GGamma \xx_1} + \sprod{\GGamma^* \xx_1,\GGamma^* \xx_0}).
    \end{align*}
    The second derivative has been calculated.
    
    We move on to the first derivative. Let us treat the  constant, linear and quadratic term of $\delta(\zz)$ separately. As for the first, it causes a term
    \begin{align*}
        \overline{\alpha}(\beta(\xx_0) + \gamma(\zz,\xx_0)).
    \end{align*}
    Here, only the constant term survives calculating the expected value. The linear part induces a term of the form
    \begin{align*}
        \overline{\beta(\zz)}(\beta(\xx_0) + \gamma(\zz,\xx_0)).
    \end{align*}
    The linear part again vanishes. The bilinear equals
    \begin{align*}
         \sprod{\zz,\bb} \sprod{\zz,\GGamma \xx_0} + \langle \tilde{\bb},\zz \rangle \sprod{\xx_0,\GGamma \zz} + \\
          \sprod{\zz,\bb} \sprod{\xx_0,\GGamma \zz} + \langle \tilde{\bb},\zz \rangle \sprod{\zz,\Gamma \xx_0}
    \end{align*}
    Using an argument similar to that above, we obtain that the above equals
    \begin{align*}
        \varsigma^2(\langle \tilde{\bb},\GGamma \xx_0\rangle + \sprod{\GGamma^*\xx_0,\bb}) = \varsigma \overline{\beta(\GGamma \xx_0)}
    \end{align*}
    We now only have the term associated to the quadratic term in $\delta$. It causes the term
    \begin{align*}
        \sprod{\zz,\GGamma \zz} \beta(\xx_0) + \sprod{\zz,\GGamma \zz} \gamma(\zz,\xx_0).
    \end{align*}
    Here, the cubic term vanishes in the expectation due to symmetry of $\zz$, and $\erw{\sprod{\zz,\GGamma \zz}} = \mathrm{tr}(\GGamma)=0$ due to the zero diagonal assumption. 
    
    We now move on to the final claim, namely the one about the expectation of $\pi(\zz)=\overline{\delta(\zz)}\delta(\zz)$ itself. We apply the same strategy as for the derivative, i.e. treat each term in $\overline{\delta(\zz)}$ separately. As for the constant part, we have
    \begin{align*}
        \overline{\alpha}\delta(\zz) = \abs{\alpha}^2 + \overline{\alpha}\beta(\zz) + \overline{\alpha}\sprod{\zz,\GGamma \zz}.
    \end{align*}
    Again, only the constant survives taking the expectation. We continue with the linear part 
    \begin{align*}
          \overline{\beta(\zz)}\delta(\zz) = \overline{\beta(\zz)}\alpha + \abs{\beta(\zz)}^2 + \overline{\beta(\zz)}\sprod{\zz,\GGamma \zz}.
    \end{align*}
    The linear and third order terms vanish in expectation due to symmetry. As for the final term, we have
    \begin{align*}
        \erw{\abs{\beta(\zz)}^2} =& \mathbb{E}\big(\sprod{\zz,\bb}\sprod{\bb,\zz} + \langle{\tilde{\bb},\zz}\rangle\sprod{\bb,\zz} \\
        & + \sprod{\zz,\bb}\langle{\zz,\tilde{\bb}}\rangle + \langle{\tilde{\bb},\zz}\rangle\langle{\zz,\tilde{\bb}}\rangle\big) \\
        &= \varsigma{\norm{\bb}^2 + \norm{\tilde{\bb}}^2}.
    \end{align*}
    The quadratic term is left. We have
    \begin{align*}
         \overline{\sprod{\zz,\GGamma \zz}}\delta(\zz) = \overline{\sprod{\zz,\GGamma  \zz}}\alpha + \sprod{\zz,\GGamma Z}\beta(\zz) + \abs{\sprod{\zz,\GGamma \zz}}^2.
    \end{align*}
    The expectation of the two first terms vanish -- the first due to $\tr(\GGamma)=0$, and the second due to symmetry. Let's expand the third one
    \begin{align*}
        \abs{\sprod{\zz,\GGamma \zz}}^2 = \sum_{i,j,k,\ell} \overline{\Gamma_{ij}}\Gamma_{k\ell} z_i \overline{z_j} \overline{z_k}z_\ell
    \end{align*}
    The expectation of $z_i \overline{z_j} z_k\overline{z_\ell}$ is zero unless either $i=j$ and $k=\ell$ or $i=k$ and $j=\ell$. In the first cases, $\overline{\Gamma_{ij}}\Gamma_{k\ell}=0$. The same thing happens when the common value of $i$ and $k$ is the same as the common value of $j$ and $\ell$. When $i=k \neq j=\ell$, we have $\erw{z_i \overline{z_j} \overline{z_k}z_\ell} = \erw{\abs{z_i}^2}\erw{ \abs{z_j}^2}=\varsigma^4$ Therefore, the only terms that survives taking the expected value is
    \begin{align*}
        \sum_{i\neq j} \abs{\Gamma_{ij}}^2 \varsigma^4 = \norm{\GGamma}_F^2\varsigma^4.
    \end{align*}
    The proof is finished.

\end{proof}

We draw the following immediate corollary.

\begin{corollary} \label{cor:derivs} Let $\alpha_{p,q}$, $\beta_{p,q}$ and $\GGamma_{p,q}$ be scalars, linear forms and matrices as in Lemma \ref{lem:derivs}. Consider the polynomial
    \begin{align*}
    \varpi(\zz)=  \sum_{p,q} \pi_{p,q}(\zz),
\end{align*}
where  $\pi_{p,q}$ is as in Lemma \ref{lem:derivs}. With the notation
\begin{align*}
    \norm{\alpha}^2 &= \sum_{p \neq q} \abs{\alpha_{p,q}}^2, \quad \norm{\beta}^2 = \left(\sum_{p\neq q} \norm{\beta_{p,q}}_F^2\right)^{\sfrac{1}{2}} \\
    \norm{\GGamma}^2 &= \left(\sum_{p\neq q} \norm{\Gamma_{p,q}}_F^2\right)^{\sfrac{1}{2}},
\end{align*}
 we have 
\begin{align*}
    \erw{\varpi(\zz)} &= \norm{\alpha}^2 + \varsigma^4 \norm{\GGamma}^2 \\
    & \quad +\varsigma^2 \sum_{p\neq q} \norm{b_{p,q}}^2 +\norm{\tilde{b}_{p,q}}^2 \\
 \norm{\erw{\varpi'(Z)}}_F &\leq 2\varsigma^2\norm{\GGamma}\norm{\beta} + 2 \norm{\alpha}\norm{\beta} \\
         \norm{\erw{\varpi''(\zz)}}_F& \leq 4\varsigma^2 \norm{\GGamma}_F^2 + 8\norm{\alpha}\norm{\GGamma} + \norm{\beta}^2 \\
         \norm{\erw{\varpi^{(3)}(\zz)}}_F& \leq 12\norm{\beta}\norm{\GGamma} \\
         \norm{\erw{\varpi^{(4)}(\zz)}}_F& \leq 24\norm{\GGamma}^2
   \end{align*} 
\end{corollary}

\begin{proof}
    Lemma \ref{lem:derivs} immediately implies the following bounds
    \begin{align*}
        \norm{\erw{\pi_{p,q}'(\zz)}}_F &\leq 2\varsigma^2\norm{\GGamma\pi_{p,q}}_F \norm{\beta_{p,q}}_F \\
        & \quad + 2\abs{\alpha_{p,q}}\norm{\beta\pi_{p,q}}_F \\
         \norm{\erw{\pi_{p,q}''(\zz)}}_F& \leq 4\varsigma^2 \norm{\Gamma_{p,q}}_F^2 \\
         & \quad + 8\abs{\alpha_{p,q}}\norm{\Gamma_{p,q}}_F + \norm{\beta_{p,q}}_F^2 \\
         \norm{\erw{\pi_{p,q}^{(3)}(\zz)}}_F& \leq 12\norm{\beta_{p,q}}\norm{\GGamma\pi_{p,q}} \\
         \norm{\erw{\pi_{p,q}^{(4)}(\zz)}}_F& \leq 24\norm{\Gamma_{p,q}}_F^2
    \end{align*}
    It is now just a matter of utilizing the linearity of the derivative and the Cauchy-Schwarz inequality to obtain the stated result.
    
\end{proof}

With the above two auxilary results in our toolbox, we may now prove Lemma \ref{lem:psivariation}.

\begin{proof}[Proof of Lemma \ref{lem:psivariation}] Notice that
\begin{align*}
    &\sprod{\AA^i(\hh_p+\zz_p^i),\BB_q^i(\hh_q+\zz_q^i)} \\
    &\quad = \sprod{\AA^i\hh_p,\BB_q^i\hh_q} + \sprod{\AA^i\hh_p,\BB_q^i\zz_q^i}  \\
    &\qquad +  \sprod{\AA^i\zz_p^i,\BB_q^i\hh_q}  + \sprod{\AA^i\zz_p^i,\BB_q^i\zz_q^i}  
\end{align*}
The expressions 
\begin{align*}
    X^i = \sum_{p\neq q }\abs{\sprod{\AA^i(\hh_p+\zz_p^i),\BB_q^i(\hh_q+\zz_q^i)}}^2
\end{align*}
we are trying to control are hence as in Corollary \ref{cor:derivs} with
\begin{align*}
\alpha_{p,q}^i&= \sprod{\AA^i\hh_p,\BB_q^i\hh_q} & \GGamma_{p,q}^i &= (\CC_p^i)^*\CC_q^i, \\
    b_{p,q}^i &= (\CC_p^i)^*\CC_q^i\hh_q &  \tilde{b}_{p,q}^i &= (\CC_q^i)^*\CC_p^i\hh_p,
\end{align*}
where we used the notation $\CC_p = \AA_\omega^*\AA_p$ again. Let us estimate the values of $\norm{\alpha^i}$, $\norm{\beta^i}$ and $\norm{\GGamma^i}$ in this case. To ease the notational burden slightly, let us drop the index $i$. 

\underline{$\norm{\alpha}$} Here, we just need to recognize the term from the $\psi(\AA,\BB,\hh)$-expression
\begin{align*}
    \norm{\alpha}^2 \leq \sum_{p\neq q} \abs{\sprod{\AA_p\hh_p,\BB_q\hh_q}}^2 \leq \psi(\AA,\BB,\hh).
\end{align*}

\underline{$\norm{\beta}$} Let us first notice that since $b_{p,q}$ and $\tilde{b}_{p,q}$ have disjoint supports (remember that $p\neq q$), we have
\begin{align*}
    \norm{\beta}^2 &= \sum_{p\neq q}\norm{\beta_{p,q}}_F^2 = \sum_{p\neq q}\norm{b_{p,q}}^2 + \norm{\tilde{b}_{p,q}}^2  \\
    &= \sum_{p\neq q} \norm{\CC_p^*\CC_qh_q}^2 + \norm{\CC_q^*\CC_ph_p}^2 \\
    &= 2 \sum_{p\neq q} \norm{\CC_q^*\CC_ph_p}^2.
\end{align*}
Remembering the notation $I_p$ for the $p$:th block, we have 
    \begin{align*}
        &\norm{\CC_q^*\CC_p\hh_p}^2 = \sum_{\ell \in I_q} \abs{\sprod{\AA_\omega^*\aa_\ell,\AA^*_\omega \AA_p\hh_p}}^2 \\
        &\quad = \sum_{\ell \in I_q} \abs{\sum_{j \in \omega} \sum_{k \in I_p} \sprod{\aa_\ell,\aa_j}\sprod{\aa_j,\aa_k}\hh_p(k)}^2 \\
        & \quad \leq \sum_{\ell \in I_q} \abs{\sum_{k \in I_p}\sum_{j \in \omega}  \abs{\sprod{\aa_\ell,\aa_j}\sprod{\aa_j,\aa_k}}\abs{\hh_p(k)}}^2
    \end{align*}
   Now we apply Lemma \ref{lem:omegasum} to estimate the expression within the square. 
    Consequently
    \begin{align*}
        \sum_{k \in I_p}&  \sum_{j \in \omega}\abs{\sprod{\aa_\ell,\aa_j}\sprod{\aa_j,\aa_k}}\abs{\hh_p(k)} \\
       &\leq \tfrac{\tau}{\sqrt{k_uk_s^2}}\sum_{k \in I_p \cap \omega} \abs{\hh_p(k)} \\
        &+ (\mathds{1}_\omega(\ell) \tfrac{\tau}{\sqrt{k_uk_s^2}} + \tfrac{\tau^2}{k_uk_s} )\sum_{k \in I_p} \abs{\hh_p(k)} \\
        &\leq \tfrac{\tau}{\sqrt{k_uk_s}}\norm{(\hh_p)_\omega} +  (\mathds{1}_\omega(\ell) \tfrac{\tau}{\sqrt{k_uk_s}} + \tfrac{\tau^2}{k_uk_s^{\sfrac{1}{2}}}) \norm{\hh_p}
    \end{align*}
    Squaring this, utilizing the inequality $(a+b+c)^2\leq 3a^2+3b^2+3c^2$ yields
    \begin{align*}
        \abs{\sprod{\AA_\omega^*\aa_\ell,\AA_\omega^*\AA_p\hh_p}}^2 \leq &   \tfrac{3\tau^2}{k_uk_s}\norm{(\hh_p)_\omega}^2 + \\
        &+ (3\mathds{1}_\omega(\ell) \tfrac{3\tau^2}{k_uk_s} + \tfrac{3\tau^4}{k_u^2k_s}) \norm{\hh_p}^2 
    \end{align*}
    Summing this over $\ell$ implies
    \begin{align*}
        \norm{\CC_q^*\CC_p\hh_p}^2 \leq& \tfrac{3\tau^2 s }{k_uk_s}\norm{(\hh_p)_\omega}^2 + ( \tfrac{3\tau^2}{k_uk_s}c_q + \tfrac{3\tau^4 s}{k_u^2k_s}) \norm{\hh_p}^2,
    \end{align*}
    where $c_q$ equals $1$ if $\omega \cap I_q \neq \emptyset$, and zero otherwise.
    Again summing over $p\neq q$, we obtain the final estimate
    \begin{align*}
        \norm{\beta}^2 &\leq \tfrac{3\tau^2 n }{k_uk_s}\norm{\hh_\omega}^2 + ( \tfrac{3\tau^2}{k_uk_s} + \tfrac{3\tau^4 n}{k_u^2k_s}) \norm{\hh}^2 \\
        &\leq n \cdot \tfrac{3\tau^2}{k_uk_s}\left(\norm{\hh_\omega}^2 + ( \tfrac{1}{n} + \tfrac{\tau^2}{k_u})\norm{\hh}^2\right) \\
        &\leq 3n k_s^{-1}\otau^2 \mathds{h}^2
    \end{align*}
    Here, we used that $\omega$ is contained in only one $I_q$, so that $\sum_{q} c_q = 1$
    
    \underline{$\norm{\GGamma}$} The argument is similar to the one above. We have
    \begin{align*}
        \norm{\CC_q^*\CC_p}^2 &= \sum_{\ell \in I_q} \sum_{k \in I_p} \abs{\sprod{\AA_\omega^*\aa_\ell,\AA_\omega^*\aa_k}}^2 \\
            &\leq \sum_{\ell \in I_q} \sum_{k \in I_p} \abs{\sum_{j \in \omega} \abs{\sprod{\aa_\ell,\aa_j}\sprod{\aa_j,\aa_k}}}^2.
    \end{align*}
    Again using the squared version of the bound \eqref{eq:akbl},  we see that the above is smaller than
    \begin{align*}
       &3\sum_{\ell \in I_q} \sum_{k \in I_p} \left(\mathds{1}_\omega(k) + \mathds{1}_\omega(\ell)\right)\tfrac{\tau^2}{k_uk_s^2} + \tfrac{\tau^4}{k_u^2k_s^2} \\
       & \ \leq 3\sum_{\ell \in I_q}\left(1 + s\mathds{1}_\omega(\ell)\right)\tfrac{\tau^2}{k_uk_s^2} + \tfrac{s\tau^4}{k_u^2k_s^2}   \leq \tfrac{6s\tau^2}{k_uk_s^2} + \tfrac{3s^2\tau^4}{k_u^2k_s^2}.
    \end{align*}
    Summing this bound over $p \neq q$, we obtain
    \begin{align*}
        \norm{\GGamma}^2 &\leq \tfrac{6nn_s\tau^2}{k_uk_s^2} + \tfrac{3n^2\tau^4}{k_u^2k_s^2}        \leq n^2 \tfrac{\tau^2}{k_uk_s} (\tfrac{1}{k_s} + \tfrac{\tau^2}{k_uk_s}) \\
        & \leq n^2 k_s^{-2}\otau^2 (1+\otau^2) \leq n^2k_s^{-2}\otau^2(1+\otau)^2
    \end{align*}
  
  We can now bound the expressions $\lambda_k$ from Theorem \ref{th:polyconc}: Remember that in this case, $\varsigma^2 = \tfrac{\sigma^2m}{n^2}$
  \begin{align*}
      \lambda_1 \leqsim&  \tfrac{\sigma^3m^{\sfrac{3}{2}}}{n^3} n k_s^{-1}\otau(1+\otau) k_s^{\sfrac{-1}{2}}n^{\sfrac{1}{2}}\otau\mathds{h} \\
      &+ \tfrac{\sigma m^{\sfrac{1}{2}}}{n}\psi(\AA,\BB,\hh)^{\sfrac{1}{2}}n^{\sfrac{1}{2}}k_s^{-\sfrac{1}{2}}\otau \mathds{h} \\
      \leq & k_s^{-\sfrac{1}{2}}(k_s^{-1}\osigma^3(1+\otau)\otau  + \osigma \psi(\AA,\BB,\hh)^{\sfrac{1}{2}}\otau) \mathds{h} \\
            \lambda_2 \leqsim& \tfrac{\sigma^4m^2}{n^4}n^2k_s^{-2}(1+\otau)^2\otau^2+ \tfrac{\sigma^2m}{n^2}nk_s^{-1}\otau^2\mathds{h}^2 \\
            &+ \tfrac{\sigma^2m}{n^2}\psi(\AA,\BB,\hh)^{\frac{1}{2}}nk_s^{-1}\otau(1+\otau) \\
            =  &\osigma^4 k_s^{-2}(1+\otau)^2\otau^2 + k_s^{-1}\osigma^2\otau^2\mathds{h}^2 \\
            &+\osigma^2\psi(\AA,\BB,\hh)^{\sfrac{1}{2}}k_s^{-1}\otau(1+\otau) \\
            \lambda_3 \leqsim & \tfrac{\sigma^3m^{\sfrac{3}{2}}}{n^3}n^{\sfrac{1}{2}}k_s^{-\sfrac{1}{2}}\otau \mathds{h}\cdot nk_s^{-1}\otau(1+\otau) \\
            &= \osigma^3k_s^{-\sfrac{3}{2}}\mathds{h}\otau(1+\otau) \\
            \lambda_4 \leqsim & \tfrac{\sigma^4m^2}{n^4}n^2k_s^{-2}\otau^2(1+\otau)^2 = \osigma^4k_s^{-2}\otau^2(1+\otau)^2
  \end{align*}
  
  We are now in a position to apply Theorem \ref{th:polyconc}. Let us define $\theta$ as
  \begin{align*}
      \min((\lambda_3 + \lambda_4)\log(tr\nicebinom)^2, (\lambda_1 + \lambda_2)\log(tr\nicebinom)).
  \end{align*}
  Then
  \begin{align*}
      &\min_{1\leq k \leq 4} \min_{1\leq j \leq k}\left( \tfrac{\theta}{\lambda_k}\right)^{\sfrac{2}{j}} 
       \\ & \quad \geqsim\min(\log(tr\nicebinom)^2,\log(tr\nicebinom), \\ &\phantom{\min(}\log(tr\nicebinom)^{\tfrac{4}{3}},\log(tr\nicebinom)^{\tfrac{4}{4}}) \\
      & \quad \geqsim \log(tr\nicebinom),
  \end{align*}
  and, 
  \begin{align*}
      \theta \leqsim &\min(\log(tr\nicebinom), k_s^{-1}\log(tr\nicebinom)^2)\cdot \\ &\big(\psi(\AA,\BB,\hh)^{\sfrac{1}{2}}(\osigma \mathds{h}\otau + \osigma^2\otau(1+\otau) \\
      & \quad+ \osigma^2(\sigma^2(1+\otau)^2\otau^2 + \osigma \otau(1+\tau)\mathds{h} + \otau^2\mathds{h}^2 \big) \\
      &\leq  \Delta
  \end{align*}
  where we dropped a few $k_s^{-\alpha}$-factors and used that $\log(tr\nicebinom) \leqsim k_s^{-1}\log(tr\nicebinom)^2$ to make the expression a bit more tidy -- this surely only makes the expression larger.
    Consequently, Theorem \ref{th:polyconc} implies
  \begin{align*}
      &\abs{X -\erw{X}} \leq \Delta
  \end{align*}
  with a failure probability smaller than 
  \begin{align*}
      &\exp( - \kappa \min_{1\leq k \leq 4} \min_{1\leq j \leq k} \left(\tfrac{\theta }{\lambda_k})\right)^{\sfrac{2}{j}}) \\
      &\quad \leq  \exp( - \kappa \log(tr\nicebinom)) \leq (tr\nicebinom)^{-\kappa},
  \end{align*} 
  where the value of the constant $\kappa$ is dependent on the implicit constant $C$ in the above. By a union bound, we get the inequality above for all times $i$ with a probability smaller than $t^{1-\kappa}r^{-\kappa}\binom{s}{k_s}^{-\kappa}$
  
  Now let us calculate $\erw{X}$. But this is easy -- by Corollary \ref{lem:derivs} and the observation $\norm{\beta_{p,q}}_F^2 = \norm{\bb_{p,q}}^2 + \norm{\tilde{\bb}_{p,q}}^2 $, we have
  \begin{align*}
      \erw{X^i} = \norm{\alpha}^2 +\tfrac{\sigma^2m}{n^2}\norm{\beta}^2 +\tfrac{\sigma^4m^2}{n^4}\norm{\GGamma}^2.
  \end{align*}
  Consequently, using our observations from above, we have
  \begin{align*}
      &\abs{\erw{X^i} - \sum_{p\neq q} \abs{\sprod{AA_p\hh_p,\BB_q\hh_q}}^2} \\
      & \quad \leq \tfrac{\sigma^2m}{n^2}\norm{\beta}^2 +\tfrac{\sigma^4m^2}{n^4}\norm{\GGamma}^2 \\
      &\quad \leqsim \tfrac{\sigma^2m}{n^2}n k_s^{-1}\otau^2\mathds{h}^2 + \tfrac{\sigma^4m^2}{n^4} n^2k_s^{-2}\otau^2 (1+\otau)^2 \\
      & \quad = \osigma^2 \otau^2\mathds{h}^2 + \osigma^4(1+\otau)^2\otau^2,
  \end{align*}
  where  we again quite crudely estimated $k_s^{-1}\leqsim 1$.
  Since $\psi(\AA,\BB,\hh+\zz)=\sum_{i \in [t]} X^i$, we obtain the claim.

\end{proof}

\section{The expressions \texorpdfstring{$\Psi(\AA,\BB,\hh)$ and $\psi(\AA,\BB,\hh)$}{Psi(A,B,h) and psi(A,b,h)}.}

Having proven the lemmata about the deviation of the $\Psi$ and $\psi$ expressions caused by the noise vector $\psi$, it is time to investigate the 'raw' expressions $\Psi(\AA,\BB,\hh)$ and $\psi(\AA,\BB,\hh)$. We start with the former.

\begin{proof}[Proof of Lemma \ref{lem:Psiraw}]
    We will prove that for each $i$,
    \begin{align*}
       \abs{ \sum_{p} \sprod{\AA_p^i\hh_p,\BB_p^i\hh_p} - \norm{\hh_\omega}^2 } &\leq \tfrac{\tau^2\norm{\hh}^2}{k_u} + \tfrac{\tau}{\sqrt{k_u}}\norm{\hh_\omega}^2 \\
    \end{align*}
    Both statements about $\Psi$ easily follows from this one -- the first one is trivial, and the second follows from routine applications of Hoeffding.  Dropping the index $i$, we however notice that  $\sprod{\AA_p\hh_p,\BB_p\hh_p}= \norm{\CC_p\hh_p}^2$, so that 
    \begin{align*}
        \abs{\norm{\hh_\omega}^2 -\Psi(\AA,\BB,\hh)} = \abs{\sum_{p \in [r]} \norm{\CC_p\hh_p}^2 - \norm{\hh_\omega}^2},
    \end{align*}
    so that everything follows directly from Lemma \ref{lem:CP}.
\end{proof}

And now for  the final lemma involving $\psi(\AA,\BB,\hh)$.

\begin{proof}[Proof of Lemma \ref{lem:psiraw}] Let us drop the index $i$ throughout the entire proof -- since we are looking for a union bound and are not aiming to apply a probabilistic argument, it will not be needed. Let us begin by investigating one term $\abs{\sprod{\AA_p\hh_p,\BB_q\hh_q}}$
\begin{align*}
    &\abs{\sprod{\AA_p\hh_p,\BB_q\hh_q}} \\
    &\quad \leq \sum_{k\in I_p, \ell \in I_q} \abs{\sprod{\AA_\omega^*\aa_k,\AA_\omega^*\aa_\ell}} \abs{\hh_p(k)\hh_q(\ell)}
\end{align*}
Lemma \ref{lem:omegasum} implies that
\begin{align*}
    \abs{\sprod{\AA_\omega^*\aa_k,\AA_\omega^*\aa_\ell}} &\leq \sum_{j \in \omega} \abs{\sprod{\aa_k,\aa_j}\sprod{\aa_j,\aa_\ell}} \\
    &\leq \left(\mathds{1}_\omega(k) + \mathds{1}_\omega(\ell)\right)\tfrac{\tau}{\sqrt{k_uk_s^2}} + \tfrac{\tau^2}{k_uk_s}.
\end{align*}
Consequently, under application of the Cauchy-Schwarz inequality
\begin{align*}
    &\sum_{k\in I_p, \ell \in I_q} \abs{\sprod{\AA_\omega^*\aa_k,\AA_\omega^*\aa_\ell}} \abs{\hh_p(k)\hh_p(\ell)}  
    \\
    &\ \leq\sum_{\substack{k\in I_p \\ \ell \in I_q}} \big(\left(\mathds{1}_\omega(k) + \mathds{1}_\omega(\ell)\right)\tfrac{\tau}{\sqrt{k_uk_s^2}} + \tfrac{\tau^2}{k_uk_s} \big)\abs{\hh_p(k)\hh_q(\ell)} \\
    &\ \leq \sum_{k\in I_p}\left(\left(\mathds{1}_\omega(k)\norm{\hh_q} + \norm{(\hh_q)_\omega}\right)\tfrac{\tau}{\sqrt{k_uk_s}} \right)\abs{\hh_p(k)} \\
    & \quad + \sum_{k\in I_p}\tfrac{\tau^2}{k_uk_s^{\sfrac{1}{2}}}\norm{\hh_q} \abs{\hh_p(k)} \\
    &\ \leq (\norm{(\hh_p)_\omega}\norm{\hh_q} + \norm{\hh_p}\norm{(\hh_q)_\omega})\tfrac{\tau}{\sqrt{k_u}} \\
    & \quad + \tfrac{\tau^2}{k_u}\norm{\hh_q}\norm{\hh_p}.
\end{align*}
Squaring this, and summing over $p$ and $q$, we obtain
\begin{align*}
    &\sum_{p \neq q} \abs{\sprod{\AA_p\hh_p,\BB_q\hh_q}}^2 \\
    &\quad \leq   \tfrac{3\tau^2}{k_u}  \sum_{p \neq q} (\norm{(\hh_p)_\omega}^2\norm{\hh_q}^2 + \norm{\hh_p}^2\norm{(\hh_q)_\omega}^2) \\
    & \qquad + \tfrac{3\tau^4}{k_u^2}\sum_{p\neq q} \norm{\hh_q}^2\norm{\hh_p}^2 \\
    &\quad \leq \tfrac{6\tau^2}{k_u}\norm{\hh_\omega}^2\norm{\hh}^2 + \tfrac{3\tau^4}{k_u^2}\norm{\hh}^4.
\end{align*}
\end{proof}

\end{document}